\documentclass[10pt,conference,compsocconf,letterpaper]{IEEEtran}
\usepackage{times}
\usepackage{xcolor}
\usepackage{datetime}
\usepackage{url}
\usepackage{url}
\usepackage[breaklinks]{hyperref}
\usepackage{breakurl}
\usepackage{verbatim}

\usepackage{amssymb,amsmath,amsthm,amsfonts,color}
\usepackage[english]{babel}
\usepackage{comment}
\usepackage[mathcal]{euscript}
\usepackage{graphicx}
\usepackage{enumerate}
\graphicspath{{images/}}
\usepackage[caption=false,font=footnotesize]{subfig}
\usepackage{array}
\usepackage{colortbl}
\usepackage{multirow}
\newcolumntype{P}[1]{>{\centering\arraybackslash}p{#1}}
\newcolumntype{M}[1]{>{\centering\arraybackslash}m{#1}}
\newcolumntype{C}[1]{>{\centering\arraybackslash}c{#1}}
\definecolor{cellgray}{RGB}{220,220,220}

\usepackage[]{algorithm}
\usepackage[noend]{algpseudocode}

\usepackage{blindtext}
\usepackage{todonotes}

\newcommand\NoThen{\renewcommand\algorithmicthen{}}

\hypersetup{
  colorlinks,
  linkcolor={red!50!black},
  citecolor={blue!50!black},
  urlcolor={blue!80!black}
}

\newcommand{\ignore}[1]{}

\newcommand{\la}{\lambda}

\newcommand{\A}{\mathcal{A}}

\newcommand{\negl}{\mathrm{negl}}

\newcommand{\Cl}{\mathsf{Clock}}

\newcommand{\Com}{\mathsf{Com}}

\newcommand{\rec}{\mathsf{Rec}}

\newcommand{\set}[1]  {\left\{#1\right\}}
\newcommand{\mb}{\mathbb}

\newcommand{\mr}{\mathrm}
\newcommand{\mc}{\mathcal}
\newcommand{\tc}[1]{\tilde{\mathcal{#1}}}
\newcommand{\mbf}{\mathbf}
\newcommand{\Ch}{\mathcal{C}}

\newtheorem{theorem}{Theorem}

\newtheorem{corollary}{Corollary}

\newtheorem{definition}{Definition}

\theoremstyle{definition}
\theoremstyle{remark}

\newenvironment{boxfig}[2]{%
     \begin{figure}[h!]
     \newcommand{\FigCaption}{#1}
     \newcommand{\FigLabel}{#2}
     \begin{center}
         \begin{tabular}{@{}|@{~~}l@{~~}|@{}}
           \hline
           \rule[-1.5ex]{0pt}{1ex}\begin{minipage}[b]{.95\linewidth}
             \vspace{1ex}
             \smallskip
             }{%
           \end{minipage}\\
           \hline
         \end{tabular}
      \caption{\FigCaption}
     \end{center}
     \vspace{-0.5cm}
   \end{figure}
}

\newif\ifextended
\extendedtrue

\begin{document}
\title{D-DEMOS: A distributed, end-to-end verifiable, internet voting system} 

\author{
  \IEEEauthorblockN
  {
    Nikos Chondros\IEEEauthorrefmark{1},
    Bingsheng Zhang\IEEEauthorrefmark{2},
    Thomas Zacharias\IEEEauthorrefmark{1},
    Panos Diamantopoulos\IEEEauthorrefmark{1},
    Stathis Maneas\IEEEauthorrefmark{3},\\
    Christos Patsonakis\IEEEauthorrefmark{1},
    Alex Delis\IEEEauthorrefmark{1},
    Aggelos Kiayias\IEEEauthorrefmark{1} and
    Mema Roussopoulos\IEEEauthorrefmark{1}
  }
  \IEEEauthorblockA{\IEEEauthorrefmark{1}
    Department of Informatics and Telecommunications,
    University of Athens,
    Greece 
  }
  \IEEEauthorblockA{\IEEEauthorrefmark{2}
    Schools of Computing and Communications, University of Lancaster, UK 
  }
  \IEEEauthorblockA{\IEEEauthorrefmark{3}
    Department of Computer Science, University of Toronto, Canada 
  }
}
\maketitle

\begin{abstract}
  E-voting systems have emerged as a powerful technology for improving democracy by reducing election cost, increasing voter participation, and even allowing voters to directly verify the entire election procedure.
Prior internet voting systems have single points of failure, which may result in the compromise of availability, voter secrecy, or integrity of the election results.

In this paper, we present the design, implementation, security analysis, and evaluation of D-DEMOS, a complete e-voting system that is distributed, privacy-preserving and end-to-end verifiable.
Our system includes a fully asynchronous vote collection subsystem that provides immediate assurance to the voter her vote was recorded as cast, without requiring cryptographic operations on behalf of the voter. 
We also include a distributed, replicated and fault-tolerant Bulletin Board component, that stores all necessary election-related information, and allows any party to read and verify the complete election process.
Finally, we also incorporate trustees, i.e., individuals who control election result production while guaranteeing privacy and end-to-end-verifiability as long as their strong majority is honest.
 
Our system is the first e-voting system whose voting operation is human verifiable, i.e., a voter can vote over the web, even when her web client stack is potentially unsafe, without sacrificing her privacy, and still be assured her vote was recorded as cast. Additionally, a voter can outsource election auditing to third parties, still without sacrificing privacy. Finally, as the number of auditors increases, the probability of election fraud going undetected is diminished exponentially.

We provide a model and security analysis of the system.
We implement a prototype of the complete system, we measure its performance experimentally, and we demonstrate its ability to handle large-scale elections.

\end{abstract}

\section{Introduction}
E-voting systems have emerged as a powerful technology to improve the election process. 
Kiosk-based e-voting systems, e.g., \cite{chaum2001surevote,chaum-esorics-2005,fisher-wote-2006,chaum2008scantegrity,benaloh2013starvote,culnane2014peered} allow the tally to be produced faster, but  require the voter's physical presence at the booth. 
Internet e-voting systems, e.g., \cite{CGS-eurocrypt-1997,adida-helios-2008,clarkson2008civitas,kutylowski2010SCV,gjosteen2013norway,zagorski2013remotegrity,chaum2001surevote,chaum2008scantegrity,zagorski2013remotegrity,DEMOS}, however, allow voters to cast their votes remotely. 
Internet voting systems have the potential to improve the democratic process by reducing election costs and by increasing voter participation for social groups that face considerable physical barriers and overseas voters. 
In addition, several internet voting systems~\cite{adida-helios-2008,kutylowski2010SCV,zagorski2013remotegrity,DEMOS} allow voters and auditors to directly verify the integrity of the entire election process, providing \emph{end-to-end verifiability}.
This is a highly desired property that has emerged in the last decade, where voters can be assured that no entities, even the election authorities, have manipulated  the election result.
Despite their potential, existing internet voting systems suffer from single points of failure, which may result in the compromise of voter secrecy, service availability, or integrity of the result~\cite{chaum2001surevote,chaum-esorics-2005,fisher-wote-2006,chaum2008scantegrity,benaloh2013starvote,CGS-eurocrypt-1997,adida-helios-2008,clarkson2008civitas,kutylowski2010SCV,gjosteen2013norway,zagorski2013remotegrity,DEMOS}.
\\\indent
In this paper, we present the design and prototype implementation of D-DEMOS, a distributed, end-to-end verifiable internet voting system, with no single point of failure during the election process (that is, besides setup). 
We set out to overcome two major limitations in existing internet voting systems. 
The first, is their dependency on centralized components. 
The second is their requirement for the voter to run special software on their devices, which processes cryptographic operations. 
Overcoming the latter allows votes to be cast with a greater variety of client devices, such as feature phones using SMS, or untrusted public web terminals.
Our design is inspired by the novel approach proposed in~\cite{DEMOS}, where the voters are used as a source of randomness to challenge the zero-knowledge protocols, which are used to enable end-to-end verifiability.
\\\indent
We design a distributed vote collection subsystem that is able to collect votes from voters and assure them their vote was recorded as cast, without requiring any cryptographic operation from the client device. This allows voters to vote via SMS, a simple console client over a telnet session, or a public web terminal, while preserving their privacy.
At election end time, vote collectors agree on a single set of votes asynchronously, and upload it to a second distributed component, the Bulletin Board. This is a replicated service that publishes its data immediately and makes it available to the public forever.
Our third distributed subsystem, \emph{trustees} are a set of persons entrusted with secret keys that can unlock information from the bulletin board. We share these secret keys among them, making sure that an honest majority is required to uncover information from the BB. Trustees interact with the BB once the votes are uploaded to it, to produce and publish the final election tally.
\\\indent
The resulting voting system is end-to-end verifiable, by the voters themselves, as well as third-party auditors; all this while preserving voter privacy. 
A voter can provide an auditor information from her ballot; the auditor can read from the distributed BB and verify the complete process, including the correctness of the election setup by election authorities. 
Additionally, as the number of auditors increases, the probability of election fraud going undetected diminishes exponentially.
\\\indent
Finally, we implement a prototype of the complete D-DEMOS voting system. 
We measure its performance experimentally, under a variety of election settings, demonstrating its ability to handle thousands of concurrent connections, and thus manage large-scale elections.
\\\indent
To summarize, we make the following contributions:
\begin{itemize}
 \item We present the world's first complete, state-of-the-art, end-to-end verifiable, distributed voting system with no single point of failure besides setup. 
 \item The system allows voters to verify their vote was tallied-as-intended without the assistance of special software or trusted devices, and external auditors to verify the correctness of the election process. 
 Additionally, the system allows voters to delegate auditing to a third party auditor, without sacrificing their privacy.
 \item We provide a model and a security analysis of our voting system.
 \item We implement a prototype of the integrated system, measure its performance and demonstrate its ability to handle large scale elections.
\end{itemize}

\section{Related work}
\noindent\textbf{Voting systems.}
Several end-to-end verifiable e-voting systems have been introduced, e.g. the kiosk-based systems~\cite{chaum-esorics-2005,fisher-wote-2006,chaum2008scantegrity,benaloh2013starvote} and the internet voting systems~\cite{adida-helios-2008,kutylowski2010SCV,zagorski2013remotegrity,DEMOS}. 
In all these works, the Bulletin Board (BB) is a single point of failure and has to be trusted.
\\ \indent 
Dini presents a distributed e-voting system, which however is not end-to-end verifiable~\cite{dini2003secure}.  
In~\cite{culnane2014peered}, there is a distributed BB  implementation, also handling vote collection, according to the design of the vVote end-to-end verifiable e-voting system~\cite{culnane2014vVote}, which in turn is an adaptation of the Pr{\^{e}}t {\`{a}} Voter e-voting system~\cite{chaum-esorics-2005}.
In~\cite{culnane2014peered}, the proper operation of the BB during ballot casting requires a trusted device for signature verification. 
In contrast, our vote collection subsystem is done so that correct execution of ballot casting can be ``human verifiable'', i.e., by simply checking the validity of the obtained receipt. 
Additionally, our vote collection subsystem is fully asynchronous, always deciding with exactly $n-f$ inputs, while in~\cite{culnane2014peered}, the system uses a synchronous approach based on the FloodSet algorithm from~\cite{Lynch:1996:DA} to agree on a single version of the state.
\\\indent
DEMOS~\cite{DEMOS} is an end-to-end verifiable e-voting system, which introduces the novel idea 
of extracting the challenge of the zero-knowledge proof protocols from the voters' random choices; we leverage this idea in our system too.
However, DEMOS uses a centralized Election Authority (EA), which maintains all secrets throughout the entire election procedure, collects votes, produces the result and commits to verification data in the BB.
Hence, the EA is a single point of failure, and because it knows the voters' votes, it is also a critical privacy vulnerability.
In this work, we address these issues by introducing distributed components for vote collection and result tabulation, and we do not assume any trusted component during election.  
Additionally, DEMOS does not provide any recorded-as-cast feedback to the voter, whereas our system includes such a mechanism.
\ifextended
Besides, DEMOS encodes the $i$-th option to $N^{i-1}$, where $N$ is greater than the total number of voters, and this option encoding has to fit in the message space of commitments. 
Therefore, the size of the underlying elliptic curve grows linearly with the number of options, which makes DEMOS not scalable with respect to the number of options. 
In this work, we overcome this problem by using a different scheme for option encoding commitments. 
Moreover, the zero-knowledge proofs in DEMOS have a big soundness error, and it decreases the effectiveness of zero-knowledge application; whereas, in our work, we obtain nearly optimal overall zero-knowledge soundness.
\fi
\\\indent
Furthermore, none of the above works provide any performance evaluation results.

\par\noindent
\textbf{State Machine Replication.}
Castro et al.~\cite{castro-osdi-1999} introduce a practical Byzantine Fault Tolerant replicated state machine protocol. 
In the last several years, several protocols for Byzantine Fault Tolerant state machine replication have been introduced to improve performance (\cite{cowling2006hq,kotla2007zyzzyva}), robustness (\cite{aublin2013rbft, clement2009making}), or both (\cite{clement2009upright}). 
Our system does not use the state machine replication approach, as it would be more costly. 
Each of our vote collection nodes can validate the voter's requests on its own. 
In addition, we are able to process multiple different voters' requests concurrently, without enforcing the total ordering inherent in replicated state machines. Finally, we do not want voters to use special client-side software to access our system.

\section{System description}
\subsection{Problem definition and goals}
We consider an \emph{election} with a single \emph{question} and $m$ \emph{options}, for a voter population of size $n$, where voting takes place between a certain \emph{begin} and \emph{end} time (the \emph{voting hours}), and each voter may select a single \emph{option}.

Our major goals in designing our voting system are three. 
First, it has to be end-to-end verifiable, so that anyone can verify the complete election process. 
Additionally, voters should be able to outsource auditing to third parties, without revealing their voting choice.
Second, it has to be fault-tolerant, so that an attack on system availability and correctness is hard. 
Third, voters should not have to trust the terminals they use to vote, as they may be malicious; voters should be assured their vote was recorded, without disclosing any information on how they voted to the malicious entity controlling their device.

%
%

\subsection{System overview}
\label{sec:sysoverview}
We employ an election setup component in our system, which we call the Election Authority (EA), to alleviate the voter from employing any cryptographic operations. 
The EA is tasked to initialize all remaining system components, and then gets destroyed to preserve privacy.
The \emph{Vote Collection} (VC) subsystem collects the votes from the voters during election hours, and assures them their vote was \emph{recorded-as-cast}.
Our \emph{Bulletin Board} (BB) subsystem, which is a public repository of all election related information, is used to hold all ballots, votes, and the result, either in encrypted on plain form, allowing any party to read from it and verify the complete election process.
The VC subsystem uploads all votes to BB at election end time.
Finally, our design includes \emph{trustees}, who are persons entrusted with managing all actions needed until result tabulation and publication, including all actions supporting end-to-end verifiability.
Trustees hold the keys to uncover any information hidden in the BB, and we use threshold cryptography to make sure a malicious minority cannot uncover any secrets or corrupt the process.



Our system starts with the EA generating initialization data for every component of our system. 
The EA encodes each election option, and \emph{commits} to it using a commitment scheme, as described below.
It encodes the $i$-th option as $\vec{e}_i$, a unit vector where the $i$-th element is $1$ and the remaining elements are $0$. 
The commitment of an option encoding is a vector of (lifted) ElGamal ciphertexts~\cite{ElGamal} over elliptic curve, that element-wise encrypts a unit vector. Note that this commitment scheme is also additively homomorphic, i.e. the commitment of $e_a + e_b$ can be computed by component-wise multiplying the corresponding commitments of $e_a$ and $e_b$.
The EA then creates a $\mathsf{vote code}$ and a $\mathsf{receipt}$ for each option.
Then, the EA prepares one ballot for each voter, with two functionally equivalent parts. 
Each part contains a list of options, along with their corresponding vote codes and receipts.
We consider ballot distribution to be outside the scope of this project, but we do assume ballots, after being produced by the EA, are distributed in a secure manner to each voter; thus only each voter knows the vote codes listed in her ballot.
We make sure vote codes are not stored in clear form anywhere besides the voter's ballot.

Our \emph{Vote Collection} (VC) subsystem collects the votes from the voters during election hours, by accepting up to one vote code from each voter. 
Our EA initializes each VC node with the vote codes and the receipts of the voters' ballots. 
However, it hides the vote codes, using a simple commitment scheme based on symmetric encryption of the plaintext along with a random salt value.
This way, each VC node can verify if a vote code is indeed part of a specific ballot, but cannot recover any vote code until the voter actually chooses to disclose it. 
Additionally, we secret-share each receipt across all VC-nodes using a $(N-f, N)$-VSS scheme with trusted dealer, making sure that a receipt can be recovered and posted back to the voter only when a strong majority of VC nodes participates successfully in our voting protocol.
With this design, our system adheres to the following contract with the voters: \emph{Any honest voter who receives a valid receipt from a Vote Collector node, is assured her vote will be published on the BB, and thus included in the election tally}.

The voter selects one part of her ballot at random, and posts her selected vote code to one of the VC nodes. 
When she receives a receipt, she compares it with the one on her ballot corresponding to the selected vote code.
If it matches, she is assured her vote was recorded and will be included in the election tally.
The other part of her ballot, the one not used for voting, will be used for auditing purposes.
This design is essential for verifiability, in the sense that the EA cannot predict which part a voter may use, and the unused part will betray a malicious EA with $1/2$ probability per audited ballot.

Our second distributed subsystem is the Bulletin Board (BB), which is a replicated service of isolated nodes.
Each BB node is initialized from the EA with vote codes and associated option encodings in committed form (again, for vote code secrecy), and each BB node provides public access to its stored information.
On election end time, VC nodes run our Vote Set Consensus protocol, which guarantees all VC nodes agree on a single set of voted $\langle\mathsf{serial\textrm{-}no},\mathsf{vote\textrm{-}code}\rangle$ tuples. 
Then, VC nodes upload this set to each BB node, which in turn publishes this set once enough VC nodes provide the same set.

Our third distributed subsystem is a set of trustees, who are persons entrusted with managing all actions needed until result tabulation and publication, including all actions supporting end-to-end verifiability.
Secrets that may uncover information in the BB are shared across trustees, making sure malicious trustees under a certain threshold cannot disclose sensitive information.
We use Pedersen's Verifiable linear Secret Sharing (VSS)~\cite{vss} to split the election data among the trustees. 
In a $(k,n)$-VSS, at least $k$ shares are required to reconstruct the original data, and any collection of less than $k$ shares leaks no information about the original data. 
Moreover, Pedersen's VSS is  additively homomorphic, i.e. one can compute the share of $a+b$ by adding the share of $a$ and the share of $b$ respectively.
Using this scheme allows trustees to perform homomorphic ``addition'' on the option-encodings of cast vote codes, and contribute back a share of the opening of the homomorphic ``total''. 
Once enough trustees upload their shares of the ``total'', the election tally is uncovered and published at each BB node.

Note that, to ensure voter privacy, the system cannot reveal the content inside an option encoding commitment at any point. 
However, a malicious EA might put an arbitrary value (say $9000$ votes for option $1$) inside such a commitment, causing an incorrect tally result. 
To prevent this, we utilize the Chaum-Pedersen  zero-knowledge proof~\cite{CP}, allowing the EA to show that the content inside each commitment is a valid option encoding, without revealing its actual content.
Namely, the prover uses Sigma OR proof to show that each ElGamal ciphertext encrypts either $0$ or $1$, and the sum of all elements in a vector is $1$.
Our zero knowledge proof is organized as follows.
First, the EA posts the initial part of the proofs on the BB.
During the election, each voter's A/B part choice is viewed as a source of randomness, $0/1$, and all the voters' coins are collected and used as the challenge of our zero knowledge proof.
After that, the trustees will jointly produce the final part of the proofs and post it on the BB before the opening of the tally. 
Hence, everyone can verify those proofs on the BB. 
For notation simplicity, we omit the zero-knowledge proof components in this paper and refer the interested reader to~\cite{CP} for details.

Due to this design, any voter can read information from the BB, combine it with her private ballot, and verify her ballot was included in the tally.
Additionally, any third-party can read the BB and verify the complete election process.
As the number of auditors increases, the probability of election fraud going undetected diminishes exponentially.
For example, even if only $10$ people audit, with each one having $1 \over 2$ probability to detect ballot fraud, the probability of ballot fraud going undetected is only ${1 \over 2}^{10} = 0.00097$.
Thus, even if the EA is malicious and, e.g.,  tries to point all vote codes to a specific option, this faulty setup will be detected because of the end-to-end verifiability of the complete system.

\subsection{System and threat model}
\label{subsec:threat}
We assume a fully connected network, where each node can reach any other node with which it needs to communicate. 
The network can drop, delay, duplicate, or deliver messages out of order. 
However, we assume messages are eventually delivered, provided the sender keeps retransmitting them.
For all nodes, we make no assumptions regarding processor speeds.
We assume the clocks of VC nodes are synchronized with real time; this is needed simply to prohibit voters from casting votes outside election hours. 
Besides this, we make no other timing assumptions in our system.
We assume the EA sets up the election and is destroyed upon completion of the setup, as it does not directly interact with the remaining components of the system, thus reducing the attack surface for the privacy of the voting system as a whole.
We also assume initialization data for every system component is relayed to it via untappable channels.
We assume the adversary does not have the computational power to violate any underlying cryptographic assumptions. 
To ensure liveness, we additionally assume the adversary cannot delay communication between honest nodes above a certain threshold.
We place no bound on the number of faulty nodes the adversary can coordinate, as long as the number of malicious nodes of each
subsystem is below its corresponding fault threshold. 
We consider arbitrary (Byzantine) failures, because we expect our system to be deployed across separate administrative domains. 

\par Let $N_v$, $N_v$, and $N_t$ be the number of VC nodes, BB nodes, and trustees respectively. 
The voters are denoted by $V_\ell$, $\ell=1,\ldots,n$. We assume that there exists a \emph{global clock} variable $\Cl\in\mathbb{N}$, and that every VC node, BB node and voter $X$ is equipped with an \emph{internal clock} variable $\Cl[X]\in\mathbb{N}$. We define the two following events on the clocks:  
\begin{enumerate}[(i).]
 \item The event $\mathsf{Init}(X):$ $\mathsf{Clock}[X]\leftarrow\mathsf{Clock}$, that initializes a node $X$ by synchronizing its internal clock with the global clock. 
 \item The event $\mathsf{Inc}(i):$ $i\leftarrow i+1$, that causes some clock $i$ to advance by one time unit.
\end{enumerate}

\begin{boxfig}{\label{fig:threat} The adversarial setting for the adversary $\A$ acting upon the distributed bulletin board system.}{}
  {\it \underline{The adversarial setting.}}
\begin{enumerate}
 \item The EA  initializes every VC node,BB node, trustee of the D-DEMOS system by running $\mathsf{Init}(\cdot)$ in all clocks for synchronization. Then, EA prepares the voters' ballots and all the VC nodes', BB nodes', and trustees' initialization data. Finally, it forwards the ballots for ballot distribution to the voters $V_\ell$, $\ell=1,\ldots,n$.
 \item $\A$ corrupts a fixed subset of VC nodes, a fixed subset of BB nodes, and a fixed subset of trustees. In addition, it defines a dynamically updated subset of voters $\mathbf{V}_\mathsf{corr}$, initially set as empty. 
\item When an honest node $X$ wants to transmit a message $\mathbf{M}$ to an honest node $Y$, then it just sends $(X,\mathbf{M},Y)$ to $\A$.
 \item $\A$ may arbitrarily invoke the events $\mathsf{Inc}(\mathsf{Clock})$ or $\mathsf{Inc}(\mathsf{Clock}[X])$, for any node $X$. Moreover, $\A$ may write on the incoming network tape of any honest component node of D-DEMOS.
 \item For every voter $V_\ell$, $\A$ can choose whether it is going to include $V_\ell$ in $\mathbf{V}_\mathsf{corr}$. 
 \begin{enumerate}[(i).]
  \item If $V_\ell\in\mathbf{V}_\mathsf{corr}$, then $\A$ fully controls $V_\ell$.
  \item If $V_\ell\notin\mathbf{V}_\mathsf{corr}$, then $\A$ may initialize $V_\ell$ by running $\mathsf{Init}(V_\ell)$ only once. If this happens, then the only control of $\A$ over $V_\ell$ is $\mathsf{Inc}(\mathsf{Clock}[V_\ell])$ invocations. Upon initialization, $V_\ell$ engages in the voting protocol.
 \end{enumerate} 
\end{enumerate}
\end{boxfig}
All system particpants are aware of a value $T_\mathsf{end}$ such that for each node $X$, if $\mathsf{Clock}[X]\geq T_\mathsf{end}$, then $X$ considers that the election has ended. 
\par The adversarial setting for $\A$ upon D-DEMOS is defined in Figure~\ref{fig:threat}. The description in Figure~\ref{fig:threat} poses no restrictions on the control the adversary has over all internal clocks, or the number of nodes that it may corrupt (arbitrary denial of service attacks or full corruption of D-DEMOS nodes are possible). Therefore, it is necessary to strengthen the model so that we can perform a meaningful security analysis and prove the properties (liveness, safety, end-to-end verifability, and voter privacy) that D-DEMOS achieves in Section~\ref{sec:security_full}. Namely, we require the following: \\
\subsubsection{Fault tolerance}\label{subsubsec:tolerance}
We consider arbitrary (Byzantine) failures. For each of the subsystems, we have the following fault tolerance thresholds:

Let $N_v$, $N_v$, and $N_t$ be the number of VC nodes, BB nodes, and trustees respectively. 
\ifextended
The voters are denoted by $V_\ell$, $\ell=1,\ldots,n$. 
For each of the subsystems, we have the following fault tolerance thresholds:
\begin{itemize}
 \item 
The number of faulty VC nodes, $f_v$, is strictly less than $1/3$ of $N_v$, i.e., for fixed $f_v$: $$\boxed{N_v\geq3f_v+1.}$$
 \item 
The number of faulty BB nodes, $f_b$, is strictly less than $1/2$ of $N_v$, i.e., for fixed $f_b$: $$\boxed{N_b\geq2f_b+1.}$$
 \item 
For the trustees' subsystem, we apply $h_t$ out-of $N_t$ threshold secret sharing, where $h_t$ is the number of honest trustees, thus we tolerate $f_t=N_t-h_t$ malicious trustees.  
\end{itemize}

\subsubsection{Liveness assumptions}\label{subsubsec:assumptions}
Only for the liveness of our system, we need to ensure eventual message delivery and bounded synchronization loss. Therefore, it is necessary, to make the following assumptions: \\
\indent\textbf{Assumption I: }There exists an upper bound $\delta$ on the time that $\A$ can delay the delivery of the messages between honest nodes. Formally, when the honest node $X$ sends $(X,M,Y)$ to $\A$, if the value of the global clock is $T$, then $\A$ must write $M$ on the incoming network tape of $Y$ by the time that  
 $\mathsf{Clock}=T+\delta$.\\
\indent\textbf{Assumption II: }There exists an upper bound $\Delta$ of the drift of all honest nodes' internal clocks with respect to the global clock. Formally, we have that: $|\mathsf{Clock}[X]-\mathsf{Clock}|\leq\Delta$  for every node $X$, where $|\cdot|$ denotes the absolute value.

\fi


\subsection{Election Authority}
\label{sect:EA}
EA produces the initialization data for each election entity in the setup phase. 
To enhance the system robustness, we let the EA generate all the public/private key pairs for all the system components (except voters) without relying on external PKI support.
We use zero knowledge proofs to ensure the correctness of all the initialization data produced by the EA. 

\emph{Voter ballots.} 
The EA generates one ballot $\mathsf{ballot}_\ell$ for each voter $\ell$, and assigns a unique $64$-bit $\mathsf{serial\textrm{-}no}_\ell$ to it.
As shown below, each ballot consists of two parts: A and B. 
Each part contains a list of $m$ $\langle\mathsf{vote\textrm{-}code},\mathsf{option},\mathsf{receipt}\rangle$ tuples, one tuple for each election option.
The EA generates the vote-code as a $160$-bit random number, unique within the ballot, and the receipt as $64$-bit random number.

\begin{center}
\begin{normalsize}
	\begin{tabular}{|l l l |}
		\hline
		$\mathsf{serial\textrm{-}no}_\ell$ &  &\\
		\hline
		& Part A &\\
		$\mathsf{vote\textrm{-}code}_{\ell,1}$ & $\mathsf{option}_{\ell,1}$ & $\mathsf{receipt}_{\ell,1}$\\
		$\quad\vdots$ & $\quad\vdots$ & $\quad\vdots$\\
		$\mathsf{vote\textrm{-}code}_{\ell,m}$ & $\mathsf{option}_{\ell,m}$ & $\mathsf{receipt}_{\ell,m}$\\
		\hline
		& Part B &\\
		$\mathsf{vote\textrm{-}code}_{\ell,1}$ & $\mathsf{option}_{\ell,1}$ & $\mathsf{receipt}_{\ell,1}$\\
		$\quad\vdots$ & $\quad\vdots$ & $\quad\vdots$\\
		$\mathsf{vote\textrm{-}code}_{\ell,m}$ & $\mathsf{option}_{\ell,m}$ & $\mathsf{receipt}_{\ell,m}$\\
		\hline
	\end{tabular}
	\end{normalsize}
\end{center}

\emph{BB initialization data.} The initialization data for all BB nodes is identical, and each BB node publishes its initialization data immediately. 
The BB's data is used to show the correspondence between the vote codes and their associated cryptographic payload.
This payload comprises the committed option encodings, and their respective zero knowledge proofs of valid encoding (first move of the prover), as described in section~\ref{sec:sysoverview}. 
However, the vote codes must be kept secret during the election, to prevent the adversary from ``stealing'' the voters' ballots and using the stolen vote codes to vote. 
To achieve this, the EA first randomly picks a $128$-bit key, $\mathsf{msk}$, and encrypts each $\mathsf{vote\textrm{-}code}$ using AES-128-CBC with random initialization vector (AES-128-CBC\$) encryption, denoted as $[\mathsf{vote\textrm{-}code}]_{\mathsf{msk}}$.
Each BB is given $H_{\mathsf{msk}}\leftarrow SHA256(\mathsf{msk} , \mathsf{salt}_\mathsf{msk})$ and $\mathsf{salt}_\mathsf{msk}$, where $\mathsf{salt}_\mathsf{msk}$ is a fresh $64$-bit random salt. Hence, each BB node can be assured the key it reconstructs from VC key-shares (see below) is indeed the key that was used to encrypt these vote-codes. 
\indent The rest of the BB initialization data is as follows: for each $\mathsf{serial\textrm{-}no}_\ell$, and for each ballot part, there is a \emph{shuffled} list of $\left\langle[\mathsf{vote\textrm{-}code}_{\ell,\pi_\ell^X(j)}]_{\mathsf{msk}},\mathsf{payload}_{\ell,\pi_\ell^X(j)} \right\rangle$ tuples, where $\pi_\ell^X\in S_L$ is a random permutation ($X$ is $A$ or $B$).
\ifextended
\begin{center}
\begin{normalsize}
	\begin{tabular}{|l c|}
		\hline
		\hspace{55pt}$(H_{\mathsf{msk}}, \mathsf{salt}_\mathsf{msk})$  &\\
		\hline
		\hline
		$\mathsf{serial\textrm{-}no}_\ell$  &\\
		\hline
		& Part A \\
		$[\mathsf{vote\textrm{-}code}_{\ell,\pi_\ell^A(1)}]_{\mathsf{msk}}$ & $\mathsf{payload}_{\ell,\pi_\ell^A(1)}$  \\
		\quad\quad $\vdots$ & $\vdots$\\
		$[\mathsf{vote\textrm{-}code}_{\ell,\pi_\ell^A(m)}]_{\mathsf{msk}}$ & $\mathsf{payload}_{\ell,\pi_\ell^A(m)}$  \\
		\hline
		& Part B \\
		$[\mathsf{vote\textrm{-}code}_{\ell,\pi_\ell^B(1)}]_{\mathsf{msk}}$ & $\mathsf{payload}_{\ell,\pi_\ell^B(1)}$  \\
		\quad\quad $\vdots$ & $\vdots$\\
		$[\mathsf{vote\textrm{-}code}_{\ell,\pi_\ell^B(m)}]_{\mathsf{msk}}$ & $\mathsf{payload}_{\ell,\pi_\ell^B(m)}$  \\
		\hline
	\end{tabular}
\end{normalsize}	
\end{center}
\fi
We shuffle the list of tuples of each part to ensure voter's privacy. 
This way, nobody can guess the voter's choice from the position of the cast vote-code in this list.

\emph{VC initialization data.} 
The EA uses an $(N_v-f_v,N_v)$-VSS to split $\mathsf{msk}$ and every $\mathsf{receipt}_{\ell,j}$, denoted as $(\|\mathsf{msk}\|_1,\ldots, \|\mathsf{msk}\|_{N_v})$ and $(\|\mathsf{receipt}_{\ell,j}\|_1,\ldots, \allowbreak\|\mathsf{receipt}_{\ell,j}\|_{N_v})$. 
For each $\mathsf{vote\textrm{-}code}_{\ell,j}$ in each ballot, the EA also computes $H_{\ell,j}\leftarrow SHA256(\mathsf{vote\textrm{-}code}_{\ell,j} , \mathsf{salt}_{\ell,j})$, where $\mathsf{salt}_{\ell,j}$ is a $64$-bit random number. 
$H_{\ell,j}$ allows each VC node to validate a $\mathsf{vote\textrm{-}code}_{\ell,j}$ individually (without network communication), while still keeping the $\mathsf{vote\textrm{-}code}_{\ell,j}$ secret. 
To preserve voter privacy, these tuples are also shuffled using $\pi_\ell^X$.
The initialization data for $VC_i$ is structured as below:
\begin{center}
\begin{normalsize}
	\begin{tabular}{|l c|}
	    \hline
	    \hspace{55pt}$\|\mathsf{msk}\|_i$  &\\
	    \hline
		\hline
		$\mathsf{serial\textrm{-}no}_\ell$ &  \\
		\hline
                & Part A \\
		$(H_{\ell,\pi_{\ell}^A(1)},\mathsf{salt}_{\ell,\pi_{\ell}^A(1)})$  & $\|\mathsf{receipt}_{\ell,\pi_{\ell}^A(1)}\|_i$ \\
		\quad $\vdots$ & $\vdots$ \\
	$(H_{\ell,\pi_{\ell}^A(m)},\mathsf{salt}_{\ell,\pi_{\ell}^A(m)})$ & $\|\mathsf{receipt}_{\ell,\pi_{\ell}^A(m)}\|_i$\\
                \hline
                & Part B \\
                $(H_{\ell,\pi_{\ell}^B(1)},\mathsf{salt}_{\ell,\pi_{\ell}^B(1)})$  & $\|\mathsf{receipt}_{\ell,\pi_{\ell}^B(1)}\|_i$ \\
                \quad $\vdots$ & $\vdots$ \\
        $(H_{\ell,\pi_{\ell}^B(m)},\mathsf{salt}_{\ell,\pi_{\ell}^B(m)})$ & $\|\mathsf{receipt}_{\ell,\pi_{\ell}^B(m)}\|_i$\\
		\hline
	\end{tabular}
	\end{normalsize}
\end{center}

\emph{Trustee initialization data.} The EA uses $(h_t,N_t)$-VSS to split the opening of encoded option commitments $\Com(\vec{e}_i)$, denoted as $(\left[\underline{\vec{e}_i}\right]_1,\ldots, \allowbreak\left[\underline{\vec{e}_i}\right]_{N_t})$.
The initialization data for $\mathsf{Trustee}_i$ is structured as below:
\begin{center}
\begin{normalsize}
	\begin{tabular}{|l c|}
		\hline
		$\mathsf{serial\textrm{-}no}_\ell$  &\\
		\hline
		& Part A \\
		 $\Com(\vec{e}_{\pi_\ell^A(i)})$  & $\left[\underline{\vec{e}_{\pi_\ell^A(i)}}\right]_\ell$\\
		$\cdots$ & $\cdots$\\
		\hline
		& Part B \\
	  $\Com(\vec{e}_{\pi_\ell^B(i)})$  & $\left[\underline{\vec{e}_{\pi_\ell^B(i)}}\right]_\ell$\\
		$\cdots$ & $\cdots$ \\
		\hline
	\end{tabular}
	\end{normalsize}
\end{center}

Similarly, the state of zero knowledge proofs for ballot correctness is shared among the trustees using $(h_t,N_t)$-VSS. 
Due to space limitation, we omit the detailed description here and refer the reader to~\cite{CP}.

\subsection{Vote Collectors}\label{subsec:Vcnodes}
VC is a distributed system of $N_v$ nodes, running our \emph{voting} and \emph{vote-set consensus} protocols.
VC nodes have private and authenticated channels to each other, and a public (unsecured) channel for voters.
%
\indent 
\ifextended
The algorithms implementing our \emph{voting} protocol are presented in Algorithm~\ref{alg:VC}.
For simplicity, we present our algorithms operating for a single election.
\fi
The \emph{voting} protocol starts when a voter submits a VOTE$\langle\mathsf{serial\textrm{-}no},\mathsf{vote\textrm{-}code}\rangle$
message to a VC node.
We call this node the \emph{responder}, as it is responsible for delivering the receipt to the voter.
The VC node confirms the current system time is within the defined election hours, and locates the ballot with the specified $\mathsf{serial\textrm{-}no}$. 
It also verifies this ballot has not been used for this election, either with the same or a different vote code.
Then, it compares the $\mathsf{vote\textrm{-}code}$ against every hashed vote code in each ballot line, until it locates the correct entry.
At this point, it multicasts an ENDORSE$\langle\mathsf{serial\textrm{-}no},\mathsf{vote\textrm{-}code}\rangle$ message to all VC nodes.
Each VC node, after making sure it has not endorsed another vote code for this ballot, responds with an ENDORSEMENT$\langle\mathsf{serial\textrm{-}no},\mathsf{vote\textrm{-}code}$,$\mathsf{sig_{VC_i}}\rangle$ message, where $\mathsf{sig_{VC_i}}$ is a digital signature of the specific serial-no and vote-code, with $VC_i$'s private key.
The responder collects $N_v-f_v$ valid signatures and forms a uniqueness certificate $\mathsf{UCERT}$ for this ballot. 
It then obtains, from its local database, the $\mathsf{receipt\textrm{-}share}$ corresponding to the specific vote-code.
Then, it marks the ballot as $\mathsf{pending}$ for the specific $\mathsf{vote\textrm{-}code}$.
Finally, it multicasts a 
VOTE\_P$\langle\mathsf{serial\textrm{-}no},\mathsf{vote\textrm{-}code}, \mathsf{receipt\textrm{-}share}$, $\mathsf{UCERT}\rangle$
message to all VC nodes, disclosing its share of the receipt.
In case the located ballot is marked as $\mathsf{voted}$ for the specific $\mathsf{vote\textrm{-}code}$, the VC node sends the stored $\mathsf{receipt}$ to the voter without any further interaction with other VC nodes.
\\\indent 
Each VC node that receives a VOTE\_P message, first verifies the validity of $\mathsf{UCERT}$, and validates the received $\mathsf{receipt\textrm{-}share}$ according to the verifiable secret sharing scheme used.
Then, it performs the same validations as the responder, and multicasts another VOTE\_P message (only once), disclosing its share of the receipt. 
When a node collects $h_v = N_v - f_v$ valid shares, it uses the verifiable secret sharing reconstruction algorithm to reconstruct the receipt (the secret) and marks the ballot as $\mathsf{voted}$ for the specific $\mathsf{vote\textrm{-}code}$. 
Additionally, the \emph{responder} node sends this receipt back to the voter.
\\\indent 
The formation of a valid $\mathsf{UCERT}$ gives our algorithms the following guarantees:
\begin{enumerate}[a)]
\item No matter how many responders and vote codes are active at the same time for the same ballot, if a $\mathsf{UCERT}$ is formed for vote code $vc_a$, no other uniqueness certificate for any vote code different than $vc_a$ can be formed.
\item By verifying the $\mathsf{UCERT}$ before disclosing a VC node's receipt share, we guarantee the voter's receipt cannot be reconstructed unless a valid $\mathsf{UCERT}$ is present. 
\end{enumerate}
\indent
At election end time, each VC node stops processing ENDORSE, ENDORSEMENT, VOTE and VOTE\_P messages, and follows the \emph{vote-set consensus} protocol, by performing the following steps for each registered ballot:
\ifextended
\begin{enumerate}
\item \label{bc-step-announce} 
Send ANNOUNCE$\langle\mathsf{serial\textrm{-}no},\mathsf{vote\textrm{-}code}, \mathsf{UCERT}\rangle$ to all nodes. 
The vote-code will be \emph{null} if the node knows of no vote code for this ballot.

\item \label{bc-step-wait} 
Wait for $N_v-f_v$ such messages. 
If any of these messages contains a valid vote code $vc_a$, accompanied by a valid $\mathsf{UCERT}$, change the local state immediately, by setting $vc_a$ as the vote code used for this ballot.

\item \label{bc-step-bc} 
Participate in a Binary Consensus protocol, with the subject ``Is there a valid vote code for this ballot?''. 
Enter with an opinion of $1$, if a valid vote code is locally known, or a $0$ otherwise.

\item \label{bc-step-result-0} 
If the result of Binary Consensus is $0$, consider the ballot not voted. 

\item \label{bc-step-result-1} 
Else, if the result of Binary Consensus is $1$, consider the ballot voted. 
There are two sub-cases here: 
  \begin{enumerate}[a)]
  \item \label{bc-step-result-1-known} 
  If vote code $vc_a$, accompanied by a valid $\mathsf{UCERT}$ is locally known, consider the ballot voted for $vc_a$. 

  \item \label{bc-step-result-1-unknown} 
  If, however, $vc_a$ is not known, send a RECOVER-REQUEST$\langle\mathsf{serial\textrm{-}no}\rangle$ message to all VC nodes, wait for the first valid RECOVER-RESPONSE$\langle\mathsf{serial\textrm{-}no}, vc_a, \mathsf{UCERT}\rangle$ response, and update the local state accordingly.
  \end{enumerate}
\end{enumerate}
\fi
\indent
Steps
\ifextended
~\ref{bc-step-announce}-\ref{bc-step-wait}
\fi
ensure used vote codes are dispersed across nodes. 
Recall our receipt generation requires $N_v-f_v$ shares to be revealed by distinct VC nodes, of which at least $N_v-2f_v$ are honest. 
Note that any two $N_v-f_v$ subsets of $N_v$ have at least one honest node in common. 
Because of this, if a receipt was generated, at least one  honest node's ANNOUNCE will be processed by every honest node, and all honest VC nodes will obtain the corresponding vote code in these two steps; thus, they enter step
\ifextended
~\ref{bc-step-bc}
\fi
with an opinion of $1$.
In this case, binary consensus is guaranteed to deliver $1$ as the resulting value (because all honest nodes share the same opinion), thus safeguarding our contract against the voters.
In any case, step
\ifextended
~\ref{bc-step-bc}
\fi
guarantees all VC nodes arrive at the same conclusion, on whether this ballot is voted or not.
\\\indent
In the algorithm outlined above, the result from binary consensus is translated from $0$/$1$ to a status of ``not-voted'' or a unique valid vote code, in steps
\ifextended
\ref{bc-step-result-0}-\ref{bc-step-result-1}. 
\fi
The 
5b
case of this translation, in particular, requires additional justification. 
Assume, for example, that a voter submitted a valid vote code $vc_a$, but a receipt was not generated before election end time. 
In this case, an honest vote collector node $VC_i$ may not be aware of $vc_a$ at step
\ifextended
~\ref{bc-step-bc}
\fi
, as steps
\ifextended
~\ref{bc-step-announce}-\ref{bc-step-wait}
\fi
do not make any guarantees in this case. 
Thus, $VC_i$ may rightfully enter consensus with a value of $0$.  
However, when honest nodes' opinions are mixed, the consensus algorithm may produce any result.
In case the result is $1$, $VC_i$ will not possess the correct vote code $vc_a$, and thus will not be able to translate the result properly. 
This is what our recovery sub-protocol is designed for. 
$VC_i$ will issue a RECOVER-REQUEST multicast, and we claim that another honest node, $VC_h$ exists that \emph{possesses} $vc_a$ and \emph{replies} with it. 
The reason for the existence of an honest $VC_h$ is straightforward and stems from the properties of the binary consensus problem definition.
If all honest nodes enter binary consensus with the same opinion $a$, the result of any consensus algorithm is guaranteed to be $a$. 
Since we have an honest node $VC_i$, that entered consensus with a value of $0$, but a result of $1$ was produced, there has to exist another honest node $VC_h$ that entered consensus with an opinion of $1$. 
Since $VC_h$ is honest, it must \emph{possess} $vc_a$, along with the corresponding $\mathsf{UCERT}$ (as no other vote code $vc_b$ can be active at the same time for this ballot). 
Again, because $VC_h$ is honest, it will follow the protocol and \emph{reply} with a well formed RECOVER-REPLY. 
Additionally, the existence of $\mathsf{UCERT}$ guarantees that any malicious replies can be safely identified and discarded.
\\\indent 
At the end of this algorithm, each node submits the resulting set of \emph{voted} $\langle \mathsf{serial\textrm{-}no}, \mathsf{vote\textrm{-}code} \rangle$ tuples to each BB node, which concludes its operation for the specific election. 


\ifextended
\begin{algorithm}
\caption{Vote Collector algorithms}
\label{alg:VC}
\begin{algorithmic}[1]
\NoThen
\footnotesize
\Procedure{on VOTE}{$\mathsf{serial\textrm{-}no},\mathsf{vote\textrm{-}code}$} from $source$:
  \If{$SysTime()$ between $start$ and $end$}
    \State $b := $locateBallot($\mathsf{serial\textrm{-}no}$)
    \If{$b.\mathsf{status} == \mathsf{NotVoted}$}
      \State $l$ := ballot.VerifyVoteCode($\mathsf{vote\textrm{-}code}$)
      \If{$l \neq null$}
        \State $b.\mathsf{UCERT} := \{\}$ \Comment{Uniqueness certificate}
        \State sendAll(ENDORSE$\langle\mathsf{serial\textrm{-}no},\mathsf{vote\textrm{-}code}\rangle$)
        \State wait for $(N_v-f_v)$ valid replies, fill $b.\mathsf{UCERT}$
        \State $b.\mathsf{status} := \mathsf{Pending}$
        \State $b.\mathsf{used\textrm{-}vc} := \mathsf{vote\textrm{-}code}$
        \State $b.\mathsf{lrs} := \{\}$ \Comment{list of receipt shares}
        \State sendAll(VOTE\_P$\langle\mathsf{serial\textrm{-}no},\mathsf{vote\textrm{-}code}, l.\mathsf{share}\rangle$)
        \State wait for $(N_v-f_v)$ VOTE\_P messages, fill $b.\mathsf{lrs}$
        \State $b.\mathsf{receipt} := \rec(b.\mathsf{lrs})$
        \State $b.\mathsf{status} := \mathsf{Voted}$
        \State $send(source, b.\mathsf{receipt})$
      \EndIf
    \ElsIf{$b.\mathsf{status} == \mathsf{Voted}$ \textbf{AND} $b.\mathsf{used\textrm{-}vc} == \mathsf{vote\textrm{-}code}$}
      \State send ($source$, $\mathsf{ballot.receipt}$)
    \EndIf
  \EndIf
\EndProcedure
\item[]
%
\Procedure{on VOTE\_P}{$\mathsf{serial\textrm{-}no},\mathsf{vote\textrm{-}code}, \mathsf{share}, \mathsf{UCERT}$} from $source$:
   \If{$\textsf{UCERT}$ is not valid}
      \State return
   \EndIf
   \If{$SysTime()$ between $start$ and $end$}
      \State $b := $locateBallot($\mathsf{serial\textrm{-}no}$)
    \If{$b.\mathsf{status} == \mathsf{NotVoted}$}
       \State $l$ := ballot.VerifyVoteCode($\mathsf{vote\textrm{-}code}$)
       \If{$l \neq null$}
         \State $b.\mathsf{status} := \mathsf{Pending}$
         \State $b.\mathsf{used\textrm{-}vc} := \mathsf{vote\textrm{-}code}$
         \State $b.\mathsf{lrs}.\mathsf{Append}(\mathsf{share})$
         \State sendAll(VOTE\_P($\mathsf{serial\textrm{-}no},\mathsf{vote\textrm{-}code}, l.\mathsf{share}$) )
       \EndIf
    \ElsIf{$b.\mathsf{status} == \mathsf{Voted}$ \textbf{AND} $b.\mathsf{used\textrm{-}vc} == \mathsf{vote\textrm{-}code}$}
        \State $b.\mathsf{lrs}.\mathsf{Append}(\mathsf{share})$
        \If{size($b.\mathsf{lrs}$) $>= N_v-f_v$}
          \State $b.\mathsf{receipt} := \rec(b.\mathsf{lrs})$
          \State $b.\mathsf{status} :=  \mathsf{Voted}$
        \EndIf
    \EndIf
  \EndIf
\EndProcedure
\item[]

\Function{Ballot::VerifyVoteCode}{$\mathsf{vote\textrm{-}code}$}
  \For{$l=1$ to $\mathsf{ballot\_lines}$}
    \If{$\mathsf{lines}[l].\mathsf{hash} == h(\mathsf{vote\textrm{-}code}||\mathsf{lines}[l].\mathsf{salt})$}
      \Return $l$
    \EndIf
  \EndFor
  \Return $null$
\EndFunction
\end{algorithmic}
\end{algorithm}
\fi

\subsection{Voter}\label{subsec:voter}
We expect the voter, who has received a ballot from EA, to know the URLs of at least $f_\mathsf{v}+1$ VC nodes.
To vote, she picks one part of the ballot at random, selects the vote code representing her chosen option, and loops selecting a VC node at random and posting the vote code, until she receives a valid receipt. 
After the election, the voter can verify two things from the updated BB. 
First, she can verify her cast vote code is included in the tally set. 
Second, she can verify that the unused part of her ballot, as ``opened'' at the BB, matches the copy she received before the election started.
This step verifies that the vote codes are associated with the expected options as printed in the ballot.
Finally, the voter can delegate both of these checks to an \emph{auditor}, without sacrificing her privacy; this is because the cast vote code does not reveal her choice, and because the unused part of the ballot is completely unrelated to the used one.

\subsection{Bulletin Board}\label{sec:BB}
A BB node is a pubic repository of election-specific information. 
By definition, it can be read via a public and anonymous channel.
Writes, on the other hand, happen over an authenticated channel, implemented with PKI originating from the voting system.
BB nodes are independent from each other; a BB node never directly contacts another BB node.
Readers are expected to issue a read request to all BB nodes, and trust the reply that comes from the majority.
Writers are also expected to write to all BB nodes; their submissions are always verified, and explained in more detail below.
\\\indent 
After the setup phase, each BB node publishes its initialization data. 
During election hours, BB nodes remain inert.
After the voting phase, each BB node receives from each VC node, the final vote-code set and the shares of $\mathsf{msk}$.
Once it receives $f_v+1$ identical final vote code sets, it accepts and publishes the final vote code set.
Once it receives $N_v-f_v$ valid key shares (again from VC nodes), it reconstructs the $\mathsf{msk}$, decrypts all the encrypted vote codes in its initialization data, and publishes them.
%
%
\\\indent 
At this point, the cryptographic payloads corresponding to the cast vote codes are made available to the trustees. 
Trustees, in turn, read from the BB subsystem, perform their individual calculations and then write to the BBs; these writes are verified by the trustees keys, generated by the EA. 
Once enough trustees have posted valid data, the BB node combines them and publishes the final election result.
\\\indent
We intentionally designed our BB nodes to be as simple as possible for the reader, refraining from using a Replicated State Machine, which would require readers to run algorithm-specific software. 
The robustness of BB nodes comes from controlling all write accesses to them. 
Writes from VC nodes are verified against their honest majority threshold. 
Further writes are allowed only from trustees, verified by their keys.
\\\indent 
Finally, a reader of our BB nodes should post her read request to all nodes, and accept what the majority responds with ($f_b+1$ is enough). 
We acknowledge there might be temporary state divergence (between BB nodes), between the time a writer finishes updating one BB node, and until he updates another. 
However, given our thresholds, this should be only momentary, alleviated with simple retries.
Thus, if there is no reply backed by a clear majority, the reader should retry until there is such a reply.

\subsection{Trustees}
After the end of election, each trustee fetches all the election data from the BB subsystem and verifies the validity of the election data. 
For each ballot, there are two possible valid outcomes: i) one of the A/B parts are voted, ii) none of the A/B parts are voted. 
If both A/B parts of a ballot are marked as voted, then the ballot is considered as invalid and should be discarded. 
Similar, trustees also discard those ballots where more than maximum allowed commitments in an A/B part are marked as voted.
In case (i), for each encoded option commitment in the voted part, $\mathsf{Trustee}_\ell$ submits its corresponding share of the opening of the commitment to the BB; for each encoded option commitment in the unused part, $\mathsf{Trustee}_\ell$ computes and posts the share of the final message of the corresponding zero knowledge proof, showing the validity of those commitments; meanwhile, those commitments marked as voted are collected to a tally set $\mathbf{E}_{\mr{tally}}$. 
In case (ii), for each encoded option commitment in both parts, $\mathsf{Trustee}_\ell$ submits its corresponding share of the opening of the commitment to the BB.
Finally, denote $\mathbf{D}^{(\ell)}_{\mr{tally}}$ as $\mathsf{Trustee}_\ell$'s set of shares of option encoding commitment openings,   corresponding to the commitments in $\mathbf{E}_{\mr{tally}}$.
$\mathsf{Trustee}_\ell$ computes the opening share for $E_{\mr{sum}}$ as $T_\ell = \sum_{D\in \mathbf{D}^{(\ell)}_{\mr{tally}}}$ and then submits $T_\ell$ to each BB node.


\subsection{Auditors}
\label{sec:auditors}
Auditors are participants of our system who can verify the election process. 
The role of the auditor can be assumed by voters or any other party.
After election end time, auditors read information from the BB and verify the correct execution of the election, by verifying the following:
a) within each opened ballot, no two vote codes are the same; 
b) there are no two submitted vote codes associated with any single ballot part; 
c) within each ballot, no more than one part has been used; 
d) all the openings of the commitments are valid; 
e) all the zero-knowledge proofs that are associated with the used ballot parts are completed and valid;
\\\indent
In case they received audit information (an unused ballot part and a cast vote code) from voters who wish to delegate verification, they can also verify:
f) the submitted vote codes are consistent with the ones received from the voters; 
g) the openings of the unused ballot parts are consistent with the ones received from the voters.

\section{Security of D-Demos}\label{sec:security_full}
In this Section, we describe at length the security properties that D-DEMOS achieves under the threat model described in Section~\ref{subsec:threat}. Specifically, we show that our distributed system achieves liveness and safety, as well as end-to-end verifiability and voter privacy at the same level of~\cite{DEMOS}\footnote{In~\cite{DEMOS}, the authors use the term \emph{voter privacy/receipt-freeness}, but they actually refer to the same property.}.\\
\par  We use $m$, $n$ to denote the number of options and voters respectively. We denote by $\lambda$ the cryptographic security parameter and we write $\mathsf{negl}(\lambda)$ to denote that a function is negligible in $\lambda$. We assume the following security guarantees for the underlying cryptographic tools:
\begin{enumerate}
 \item The probability that an adversary running in  $\lambda$ steps forges digital signatures is no more than $\mathsf{negl}(\lambda)$.
 \item There exists a constant $c<1$ s.t. the probability that an adversary running in $O(2^{\lambda^c})$ steps breaks the hiding property of the option-encoding commitments is no more than $\mathsf{negl}(\lambda)$.
\end{enumerate}
 %
\subsection{Liveness}\label{subsec:sec_liveness}
%

To prove the liveness our DBB guarantees, we assume (I) an upper bound $\delta$ on the delay of the delivery of messages and (II) an upper bound $\Delta$ on the drift of all clocks (see Section~\ref{subsubsec:assumptions}). Furthermore, to express liveness rigorously, we formalize the behavior of honest voters regarding maximum waiting before vote resubmission as follows: 

\begin{definition}[$\mathrm{[}d\mathrm{]}$-\textbf{patience}]\label{dfnt:patient}
Let $V$ be an honest voter that submits her vote at some VC node when $\Cl[V]=T$. We say that $V$ is $[d]$-\emph{patient}, when the following condition holds:
\noindent If $V$ does not obtain a receipt by the time that $\Cl[V]=T+d$, then she will blacklist this VC node and submit the same vote to another randomly selected VC node.

\end{definition}

\noindent Using Definition~\ref{dfnt:patient}, we prove the liveness of our e-voting system in the following theorem:

\begin{theorem}[\textbf{Liveness}]\label{thm:liveness}
 Let $\A$ be an adversary against D-DEMOS under the model described in Section~\ref{subsec:threat}, where assumptions I,II in Section~\ref{subsubsec:assumptions} hold and $\A$ corrupts up to $f_v<N_v/3$ VC nodes.  Let $T_\mathsf{comp}$ be the worst-case running time of any procedure run by the VC nodes and the voters respectively during the voting protocol.
 \indent Let $T_\mathsf{end}$ denote the election end time. Define
 \[T_\mathsf{wait}:=(2N_v+4)T_\mathsf{comp}+12\Delta+6\delta\;.\] 
 Then, the following conditions hold:
 \begin{enumerate}
   \item\label{liveness-1} Every $[T_\mathsf{wait}]$-patient voter that is engaged in the voting protocol by the time that $\Cl[V]=T_\mathsf{end}-(f_v+1)\cdot T_\mathsf{wait}$, will obtain a valid receipt.
   \item\label{liveness-2} Every $[T_\mathsf{wait}]$-patient voter that is engaged in the voting protocol by the time that  
  $\Cl[V]=T_\mathsf{end}-y\cdot T_\mathsf{wait}$, where $y\in[f_v]$,
  will obtain a valid receipt with more than $1-3^{-y}$ probability.
 \end{enumerate}

\end{theorem}

\begin{proof}
 Let $V$ be a $[T_\mathsf{wait}]$-patient voter initialized by the adversary $\A$ when $\Cl=\Cl[V]=T$. We will compute an upper bound on the time required for an honest responder $VC$ node to issue a receipt to $V$. This bound will be derived by the time upper bounds that correspond to each step of the voting protocol, as described in Sections~\ref{subsec:Vcnodes} and~\ref{subsec:voter}, taking also into account the $\Delta$ upper bound on clock drifts and the $\delta$ upper bound on message delivery. In Table~\ref{tab:liveness}, we provide the advance of these upper bounds at the global clock and the internal clocks of $V$ and $VC$, so that we illustrate the description of the computation described below.
 \par Upon initialization, $V$'s internal clock is synchronized with the global clock at time $\Cl=\Cl[V]=T$. After at most $T_\mathsf{comp}$ steps, $V$  submits her vote $(\mathsf{serial\textrm{-}no},\mathsf{vote\textrm{-}code})$ at internal clock time: $\Cl[V]=T+T_\mathsf{comp}$, hence at global clock time: $\Cl\leq T+\Delta$. 
 \par Thus, $VC$ will receive the vote of $V$ at internal time $\Cl[VC]\leq T+T_\mathsf{comp}+2\Delta+\delta$.  Then, $VC$ performs at most $T_\mathsf{comp}$ steps to verify the validity of the vote before it broadcasts its
 an ENDORSE request to the other VC nodes by global clock time 
 $\Cl\leq T+T_\mathsf{comp}+2\Delta+\delta+(T_\mathsf{comp}+\Delta)= T+2T_\mathsf{comp}+3\Delta+\delta$.
 \par All the other honest VC nodes ($N_v-f_v-1$ in total) will receive $VC$'s ENDORSEMENT request by global clock time: $\Cl\leq T+2T_\mathsf{comp}+3\Delta+2\delta$, which implies that the time at their internal clocks is at most
  $T+2T_\mathsf{comp}+4\Delta+2\delta$.
Upon receiving the request, all the honest VC nodes will verify if $\mathsf{vote\textrm{-}code}$ has never been endorsed before and if so, they respond with an ENDORSEMENT message after at most $T_\mathsf{comp}$ steps. The global clock at that point is at most
 $\Cl=(T+2T_\mathsf{comp}+4\Delta+2\delta)+T_\mathsf{comp}+\Delta=T+3T_\mathsf{comp}+5\Delta+2\delta$. Therefore, $VC$ will obtain the other honest VC nodes' ENDORSE messages at most when $\Cl[VC]=(T+3T_\mathsf{comp}+5\Delta+2\delta)+\Delta+\delta=T+3T_\mathsf{comp}+6\Delta+3\delta$. 
 \par In order to determine the uniqueness certificate UCERT for $V$'s vote, $VC$ has to verify up to $N_v-1$ messages (as the $f_v$ malicious VC nodes may also send arbitrary messages). Since the message verification  and the UCERT generation require $(N_v-1)\cdot T_\mathsf{comp}$ and $T_\mathsf{comp}$ steps respectively , $VC$ will broadcast its receipt share at global clock time 
 $\Cl\leq(T+3T_\mathsf{comp}+6\Delta+3\delta)+(N_v-1)\cdot T_\mathsf{comp}+T_\mathsf{comp}+\Delta=T+(N_v+3)T_\mathsf{comp}+7\Delta+3\delta$.
\par All the other honest VC nodes will receive $VC$'s receipt share by global clock time: $\Cl\leq T+(N_v+3)T_\mathsf{comp}+7\Delta+4\delta$, which implies that the time at their internal clocks is at most
  $T+(N_v+3)T_\mathsf{comp}+8\Delta+4\delta$.
  Then, they will verify $VC$'s share and broadcast their shares for $V$'s vote after at most $T_\mathsf{comp}$ steps. The global clock at that point is no more than
 $\Cl=T+(N_v+3)T_\mathsf{comp}+8\Delta+4\delta+(T_\mathsf{comp}+\Delta)=T+(N_v+4)T_\mathsf{comp}+9\Delta+4\delta$.
\par Therefore, $VC$ will obtain the other honest VC nodes' shares at most when
 $\Cl[VC]=T+(N_v+4)T_\mathsf{comp}+9\Delta+4\delta+(\Delta+\delta)=T+(N_v+4)T_\mathsf{comp}+10\Delta+5\delta$ and will process them in order to reconstruct the receipt for $V$. In order to collect $N_v-f_v-1$ receipt shares that are sufficient for reconstruction, $VC$ may have to perform up to $N_v-1$ receipt-share verifications, as the $f_v$ malicious VC nodes may also send invalid messages.
This verification requires at most $(N_v-1)\cdot T_\mathsf{comp}$ steps. Taking into account the $T_\mathsf{comp}$ steps for the reconstruction process, we conclude that $VC$ will finish computation by global time
 \begin{equation*}
\begin{split}
 \Cl&=T+(N_v+4)T_\mathsf{comp}+10\Delta+5\delta+\\
 &\hspace{10pt}+(N_v-1)T_\mathsf{comp}+T_\mathsf{comp}+\Delta=\\
 &=T+(2N_v+4)T_\mathsf{comp}+11\Delta+5\delta.
\end{split}
\end{equation*}
Finally, $V$ will obtain the receipt after at most $\delta$ delay from the moment that $VC$ finishes computation. Taking into consideration the drift on $V$'s internal clock, we have that if $V$ is honest and  has not yet obtained a receipt by the time that 
   \begin{equation*}
\begin{split}\Cl[V]&=\big(T+(2N_v+4)T_\mathsf{comp}+11\Delta+5\delta\big)+\Delta+\delta=\\
&=T+(2N_v+4)T_\mathsf{comp}+12\Delta+6\delta=T+T_\mathsf{wait},
\end{split}
\end{equation*}
 
 \noindent then, being $[T_\mathsf{wait}]$-patient, she can blacklist $VC$ and resubmit her vote to another VC node. We will show that the latter fact implies conditions~\ref{liveness-1} and~\ref{liveness-2} in the statement of the theorem:\\
\par\textbf{Condition~\ref{liveness-1}: }since there are at most $f_v$ malicious VC nodes, $V$ will certainly run into an honest VC node at her $(f_v+1)$-th attempt (if reached). Therefore, if $V$ is engaged in the voting protocol by the time that $\Cl[V]=T_\mathsf{end}-(f_v+1)\cdot T_\mathsf{wait}$, then she will obtain a receipt.\\
\par\textbf{Condition~\ref{liveness-2}: } if $V$ has waited for more than $y\cdot T_\mathsf{wait}$ time and has not yet received a receipt, then it has run at least $y$ failed attempts in a row. 
  At the $j$-th attempt, $V$ has $\dfrac{f_v-(j-1)}{N_v-(j-1)}$ probability to randomly select one of the remaining $f_v-(j-1)$ malicious VC nodes out of the $N_v-(j-1)$ non-blacklisted VC nodes. Thus,
 the probability that $V$ runs at least $y$ failed attempts in a row is
  \begin{equation*}
\begin{split}
 \prod_{j=1}^{y}\dfrac{f_v-(j-1)}{N_v-(j-1)}&= \prod_{j=1}^{y}\dfrac{f_v-(j-1)}{3f_v+1-(j-1)}<3^{-y}.
\end{split}
\end{equation*}  
Therefore, if $V$ is engaged in the voting protocol by the time that  $\Cl[V]=T_\mathsf{end}-y\cdot T_\mathsf{wait}$, then the probability that she will obtain a receipt is more than $1-3^{-y}$.

 \begin{table*}[ht]
 \footnotesize
\begin{center}
 \begin{tabular}{|M{3.5cm}|M{3.0cm}|M{3.0cm}|M{3.0cm}|M{3.0cm}|}
 \hline
\multirow{2}{*}{\textbf{Step}}&\multicolumn{4}{c|}{\textbf{Time upper bounds at each clock}}\\\cline{2-5}
&$\Cl$&$\Cl[V]$&$\Cl[VC]$&honest VC nodes' clocks\\
\hline\hline
$V$ is initialized & \cellcolor{cellgray}$T$ & $T$ & $T+\Delta$& $T+\Delta$\\
\hline
$V$ submits her vote to $VC$ & $ T+T_\mathsf{comp}+\Delta$ & \cellcolor{cellgray}$T+T_\mathsf{comp}$ & $T+T_\mathsf{comp}+2\Delta$&$T+T_\mathsf{comp}+2\Delta$\\
\hline
$VC$ receives $V$'s ballot & $ \cellcolor{cellgray}T+T_\mathsf{comp}+\Delta+\delta$ & $T+T_\mathsf{comp}+2\Delta+\delta$ & $T+T_\mathsf{comp}+2\Delta+\delta$&$T+T_\mathsf{comp}+2\Delta+\delta$\\
\hline
$VC$ verifies the validity of $V$'s ballot and broadcasts an ENDORSE message & $ T+2T_\mathsf{comp}+3\Delta+\delta$ & $T+2T_\mathsf{comp}+4\Delta+\delta$ & \cellcolor{cellgray}$T+2T_\mathsf{comp}+2\Delta+\delta$&$T+2T_\mathsf{comp}+4\Delta+\delta$\\
\hline
All the other honest VC nodes receive $VC$'s ENDORSE message& \cellcolor{cellgray}$T+2T_\mathsf{comp}+3\Delta+2\delta$ & $T+2T_\mathsf{comp}+4\Delta+2\delta$ & $T+2T_\mathsf{comp}+4\Delta+2\delta$&$T+2T_\mathsf{comp}+4\Delta+2\delta$\\
\hline
All the other honest VC nodes verify the validity of the ENDORSE message and respond with an ENDORSEMENT message & $T+3T_\mathsf{comp}+5\Delta+2\delta$ & $T+3T_\mathsf{comp}+6\Delta+2\delta$ & $T+3T_\mathsf{comp}+6\Delta+2\delta$&\cellcolor{cellgray}$T+3T_\mathsf{comp}+4\Delta+\delta$\\
\hline
$VC$ receives the ENDORSEMENT messages of all the other honest VC nodes &\cellcolor{cellgray}$T+3T_\mathsf{comp}+5\Delta+3\delta$ & $T+3T_\mathsf{comp}+6\Delta+3\delta$ & $T+3T_\mathsf{comp}+6\Delta+3\delta$&$T+3T_\mathsf{comp}+6\Delta+3\delta$\\
\hline
$VC$ verifies the validity of all the $N_v-1$ received messages until it obtains $N_v-f_v$ valid ENDORSEMENT messages & $ T+(N_v+2)T_\mathsf{comp}+7\Delta+3\delta$ & $T+(N_v+2)T_\mathsf{comp}+8\Delta+3\delta$ & \cellcolor{cellgray}$T+(N_v+2)T_\mathsf{comp}+6\Delta+3\delta$& $T+(N_v+2)T_\mathsf{comp}+8\Delta+3\delta$\\
\hline
$VC$ forms UCERT certificate and broadcsts its share and UCERT &$ T+(N_v+3)T_\mathsf{comp}+7\Delta+3\delta$ & $T+(N_v+3)T_\mathsf{comp}+8\Delta+3\delta$ & \cellcolor{cellgray}$T+(N_v+3)T_\mathsf{comp}+6\Delta+3\delta$& $T+(N_v+3)T_\mathsf{comp}+8\Delta+3\delta$\\
\hline
All the other honest VC nodes receive $VC$'s broadcast share and UCERT& \cellcolor{cellgray}$ T+(N_v+3)T_\mathsf{comp}+7\Delta+4\delta$ & $T+(N_v+3)T_\mathsf{comp}+8\Delta+4\delta$ &$T+(N_v+3)T_\mathsf{comp}+8\Delta+4\delta$& $T+(N_v+3)T_\mathsf{comp}+8\Delta+4\delta$\\
\hline
All the other honest VC nodes verify the validity of UCERT and $V$'s share and broadcast their shares &$ T+(N_v+4)T_\mathsf{comp}+9\Delta+4\delta$ & $T+(N_v+4)T_\mathsf{comp}+10\Delta+4\delta$ &$T+(N_v+4)T_\mathsf{comp}+10\Delta+4\delta$&\cellcolor{cellgray} $T+(N_v+4)T_\mathsf{comp}+8\Delta+4\delta$\\
\hline
$VC$ receives all the other honest VC nodes' shares &  \cellcolor{cellgray}$ T+(N_v+4)T_\mathsf{comp}+9\Delta+5\delta$ & $T+(N_v+4)T_\mathsf{comp}+10\Delta+5\delta$ &$T+(N_v+4)T_\mathsf{comp}+10\Delta+5\delta$&$T+(N_v+4)T_\mathsf{comp}+10\Delta+5\delta$\\
\hline
$VC$ verifies the validity of all the $N_v-1$ received messages until it obtains $N_v-f_v$ valid shares &$ T+(2N_v+3)T_\mathsf{comp}+11\Delta+5\delta$ & $T+(2N_v+3)T_\mathsf{comp}+12\Delta+5\delta$ &  \cellcolor{cellgray}$T+(2N_v+3)T_\mathsf{comp}+10\Delta+5\delta$&$T+(2N_v+3)T_\mathsf{comp}+12\Delta+5\delta$\\
\hline
$VC$ reconstructs and $V$'s receipt and sends it to $V$& $ T+(2N_v+4)T_\mathsf{comp}+11\Delta+5\delta$ & $T+(2N_v+4)T_\mathsf{comp}+12\Delta+5\delta$ &  \cellcolor{cellgray}$T+(2N_v+4)T_\mathsf{comp}+10\Delta+5\delta$&$T+(2N_v+4)T_\mathsf{comp}+12\Delta+5\delta$\\
\hline
$V$ obtains her receipt &$ T+(2N_v+4)T_\mathsf{comp}+11\Delta+6\delta$ &\cellcolor{cellgray}  $T+(2N_v+4)T_\mathsf{comp}+12\Delta+6\delta$ &  $T+(2N_v+4)T_\mathsf{comp}+12\Delta+6\delta$&$T+(2N_v+4)T_\mathsf{comp}+12\Delta+6\delta$\\
\hline
 \end{tabular}\end{center}
  \caption{Time upper bounds at $\Cl,\Cl[V]$, $\Cl[VC]$ and other honest VC nodes' clocks at each step of the interaction of $V$ with responder $VC$. The grayed cells indicate the reference point of the clock drifts at each step.}
\label{tab:liveness}
 \end{table*}
\end{proof}

\subsection{Safety}\label{subsec:sec_safety}
%

Our safety theorem is stated in the form of a contract adhered by the VC subsystem. Namely, the following safety theorem proves the gurantee that the recorded-as-cast feedback of D-DEMOS provides for every honest voter that obtained a valid receipt at the voting protocol.
\begin{theorem}[Safety]\label{thm:sec_safety}
Let $\A$ be an adversary against D-DEMOS under the model described in Section~\ref{subsec:threat} that corrupts up to $f_v<N_v/3$ VC nodes, up to $f_b<N_b/2$ BB nodes and up to $N_t-h_t$ out-of $N_t$ trustees.  
Then, every honest voter who receives a valid receipt from a VC node, is assured her vote will be published on the honest BB nodes and included in the election tally, with probability at least \[1-\mathsf{negl}(\lambda)-\dfrac{f_v}{2^{-64}-f_v}\;.\]
\end{theorem}
\begin{proof}
 Let $V$ be an honest voter. Then, $\A$' strategies on attacking safety (i.e. provide a valid receipt to $V$ but force the VC subsystem to discard $V$'s ballot), is captured by either one of the two following cases:
 \begin{enumerate}
  \item[\textbf{Case 1.}] Produce the receipt without being involved in a complete interaction with the VC subsystem (i.e. with at least $f_v+1$ honest VC nodes). 
  \item[\textbf{Case 2.}] Provide a properly reconstructed receipt via a complete interaction with the VC subsystem.
 \end{enumerate}
Let $E_1$ (resp. $E_2$) be the event that Case 1 (resp. Case 2)  happens. We study both cases:\\
 \par\underline{\emph{Case 1.}} In this case, $\A$ must produce a receipt that matches $V$'s ballot with less than $N_v-f_v$ shares. $\A$ may achieve this by either one of the following ways:
 \begin{enumerate}[(i).]
  \item $\A$ attempts to guess the valid receipt; If $\A$ succeeds, then it can force the VC subsystem to consider $V$'s ballot not voted as no valid UCERT certificate will be generated for $V$'s ballot. Since there are at most $f_v$ malicious VC nodes, the adversary has at most $f_v$ attempts to guess the receipt. Moreover, the receipt is a randomly generated 64-bit string, so after $i$ attempts, $\A$ has to guess among $(2^{64}-i)$ possible choices. Therefore, the probability succeeds is no more than \[\displaystyle\sum_{i=0}^{f_v-1}\dfrac{1}{2^{-64}-i}\leq \dfrac{f_v}{2^{-64}-f_v}.\]
 \item $\A$ attempts to reconstruct the receipt using all the malicious VC nodes' receipt shares. These are at most $f_v$, so by the information theoretic security of the $(N_v-f_v,N_v)$-VSS scheme the probability of success for $\A$ is no better than above ($\A$ can only perform a random guess).
 \item $\A$ attempts to produce fake UCERT certificates by forging digital signatures of other nodes. By the security of the digital signature scheme, this attack has $\negl(\lambda)$ success probability.
 \end{enumerate}
By the above, we have that
\begin{equation}\label{eq:safety-1}
 \Pr[\A\mbox{ wins }|E_1]\leq\dfrac{f_v}{2^{-64}-f_v}+\negl(\lambda)\;.
\end{equation}
 \par\underline{\emph{Case 2.}} In this case, by the security arguments stated in Section~\ref{subsec:Vcnodes} (steps \label{bc-step-announce} -\label{bc-step-result-1} ), every honest VC node will include the vote of $V$ in the set of voted tuples. This is because a) it locally knows the valid (certified) vote-code for $V$ accompanied by UCERT or b) it has obtained the valid vote-code via a RECOVER-REQUEST message. Recall that unless there are fake certificates (which happens with negligible probability) there can be only one valid vote-code for $V$.
 \par Consequently, all the honest VC nodes will forward the agreed set of votes (hence, also $V$'s vote) to the BB nodes.
 By the fault tolerance threshold for the BB subsystem, the $f_b$ honest BB nodes will publish $V$'s vote. Finally, the $h_t$ out-of $N_t$ honest trustees will read $V$'s vote from the majority of BB nodes and include it in the election tally. Thus, we have that 
\begin{equation}\label{eq:safety-2}
 \Pr[\A\mbox{ wins }|E_2]=\negl(\lambda)\;.
\end{equation}
By Eq.~\eqref{eq:safety-1},\eqref{eq:safety-2}, if $V$ obtains a valid receipt, then his vote will be published on the honest BB nodes and included in the election tally, with probability at least \[1-\mathsf{negl}(\lambda)-\dfrac{f_v}{2^{-64}-f_v}\;.\]
\end{proof}
Theorem~\ref{thm:sec_safety} provides guarantee that the honest voter's vote will be recorded-as-cast by the system on an individual level. Using Theorem~\ref{thm:sec_safety}, we prove the following corollary about the ``universal'' safety of the receipt-based feedback mechanism of D-DEMOS.

\begin{corollary}\label{cor:safety}
 Let $n$ be the number of voters. Let $\A$ be an adversary against D-DEMOS under the model described in Section~\ref{subsec:threat} that corrupts up to $f_v<N_v/3$ VC nodes, up to $f_b<N_b/2$ BB nodes and up to $N_t-h_t$ out-of $N_t$ trustees.  
Then, the probability that $\A$ achieves in excluding the vote of at least one honest voter that obtained a valid receipt from the election tally is no more than\[\dfrac{nf_v}{2^{-64}-f_v}+\mathsf{negl}(\lambda)\;.\]
\end{corollary}
\begin{proof}
 The proof is straightforward by taking a union bound on the probability of successful attack on the safety of every single honest voter and the upper bound $\dfrac{f_v}{2^{-64}-f_v}+\mathsf{negl}(\lambda)$ on $\A$'s success probability derived by Theorem~\ref{thm:sec_safety}.
\end{proof}

\subsection{End-to-end Verifiability}\label{subsec:sec_e2e}
%

We adopt the end-to-end (E2E) verifiability definition in~\cite{DEMOS}, modified accordingly to our setting. Namely, we encode the options set $\{\mathsf{option}_1,\ldots,\mathsf{option}_m\}$, where the encoding of $\mathsf{option}_i$ is an $m$-bit string which is only in the $i$-th position. Let $F$ be the election evaluation function such that $F(\mathsf{option}_{i_1}\ldots,\mathsf{option}_{i_n})$ is equal to an $m$-vector whose $i$-th location is equal to the number of times $\mathsf{option}_i$ was voted. Then, we use the metric $\mr{d}_1(\cdot,\cdot)$ derived by the 1-norm, $\|\cdot\|_1$ scaled to half, i.e.,
\begin{center}
 \begin{tabular}{ccl}
  $\mr{d}_1: $&$ \mb{Z}_+^m\times\mb{Z}_+^m$&$\longrightarrow\mb{R}$\\
  &$(w,w')$&$\longmapsto\frac{1}{2}\cdot\sum_{i=1}^n|w_i-w'_i|$
 \end{tabular}
\end{center}
to measure the success probability of the adversary with respect to the amount of tally deviation $d$ and the number of voters that perform audit $\theta$.
\par We model end-to-end verifiability  via a game between a challenger and an adversary. The adversary starts by selecting the identities of the voters, options, VC nodes, BB nodes and trustee nodes
identities for given parameters $n,m,N_v,N_b,N_t$.
In~\cite{DEMOS}, the adversary corrupts the EA which manages the elestion setup, the vote collection, and the tally computation. Analogolously, the adversary in D-DEMOS now fully controls the EA, all the trustees, and all the VC nodes. Moreover,~\cite{DEMOS} assumes a consistent BB. In our model, we guarantee a this BB by restricting the adversary to statically corrupt a minority of the BB nodes.  
\par At the voting phase, it manages the vote casting of every voter. For each voter, the adversary may choose to corrupt her 
or to allow the challenger to play on her behalf. In the second case, the adversary
provides the option selection that the honest voter will voter. The adversary finally posts the election transcript
to the BB.
Finally, we make use of a \emph{vote extractor} algorithm $\mathcal{E}$ (not necessarily running in polynomial-time) that extracts the non-honestly cast votes. The adversary will win the game provided that there is a subset 
 of  $\theta$ honest voters that audit the result successfully but the deviation of the tally is bigger
than $d$; the adversary will also win in case the vote extractor fails to produce the option
selection of the dishonest voters but still,  $\theta$ honest voters verify correctly. 
The attack game is specified in detail in Figure~\ref{fig:int.game}. 
	\begin{boxfig}{\label{fig:int.game}The E2EVerifiability Game between the challenger $\Ch$ and the adversary $\mathcal{A}$ using the vote extractor $\mathcal{E}$.}{}	
			\underline{E2E Verifiability Game $G_{\mathsf{e2e\textrm{-}ver}}^{\mathcal{A},\mathcal{E},d,\theta}(1^\la,m,n,N_v,N_b,N_t)$}
	\begin{enumerate}
	\item[]
		\item $\mathcal{A}$ on input $1^\lambda, n, m,N_v,N_b,N_t$, chooses a list of options $\{\mathsf{option}_1,\ldots,\mathsf{option}_m\}$, a set of voters $\mathcal{V}=\{V_1,\ldots,V_n\}$, a set of VC nodes $\mathcal{VC}=\{\mathsf{VC}_1,\ldots,\mathsf{VC}_{N_v}\}$, a set of 
  BB nodes $\mathcal{BB}=\{\mathsf{BB}_1,\ldots,\mathsf{BB}_{N_b}\}$, and a set of trustees $\mathcal{T}=\{T_1,\ldots,T_{N_t}\}$. It provides the challenger $\mathsf{Ch}$ with all the above sets. 
\par$\hspace{3pt}$ Throughout the game, $\mathcal{A}$ controls the EA, all the VC nodes and  all the trustees. In addition, $\mathcal{A}$ may corrupt a fixed set of up to BB nodes, denoted by $\mathcal{BB}_\mathsf{succ}$. On the other hand, $\mathsf{Ch}$ plays the role of all the honest BB nodes.
		\item
		$\mathcal{A}$ and $\Ch$ engage in an interaction
		where $\mathcal{A}$ schedules the vote casting executions of all voters. 
		For each voter $V_\ell$, $\mathcal{A}$ can either completely control the voter 
		or allow $\Ch$ to operate on $V_\ell$'s behalf, in which case  $\mathcal{A}$ provides $\Ch$ with an option selection $\mathsf{option}_{i_\ell}$. 
		Then, $\Ch$ casts a vote for $\mathsf{option}_{i_\ell}$,  and, provided the voting execution terminates successfully,  $\Ch$ obtains
		the audit information $\mathsf{audit}_\ell$ on behalf of $V_\ell$.\vspace{1pt}
	\par$\hspace{3pt}$	Let $\mathcal{V}_\mathsf{succ}$ be the set of honest voters (i.e., those controlled by  $\Ch$) that terminated successfully.
		\item Finally, $\mathcal{A}$ posts a version of the election transcript $\mathsf{info}_j$ in every honest BB node $\mathsf{BB}_j\notin\mathcal{BB}_\mathsf{corr}$.
	\end{enumerate}   
	The game returns a bit which is $1$ if and
	only if the  following conditions hold true:
	\begin{enumerate}
	        \item[1.] $|\mathcal{BB}_\mathsf{corr}|<\lfloor N_b/2\rfloor$ (i.e., the majority of the BB nodes remain honest). 
	        \item[2.] $\forall\mathsf{BB}_{j}, \mathsf{BB}_{j'}\notin\mathcal{BB}_\mathsf{corr}:$ $\mathsf{info}_j=\mathsf{info}_{j'}:=\mathsf{info}$ (i.e, the data posted in all honest BB nodes are identical).
		\item[3.] $ | \mathcal{V}_\mathsf{succ}|  \geq \theta$ (i.e., at least $\theta$ honest voters terminated). 
		\item[4.] $\forall\ell\in[n]:$ if $ V_\ell\in\mathcal{V}_\mathsf{succ}$ then $V_\ell$ verifies succesfull, when given $(\mathsf{info},\mathsf{audit}_\ell)$ as input. 
	\end{enumerate}
	and either one of the following two conditions: 
	\begin{itemize}
		\item[5.a.]  if $\bot\neq \langle\mathsf{option}_{i_\ell}\rangle_{ V_\ell\notin\mathcal{V}_\mathsf{succ}} \leftarrow$ $\mathcal{E}(\mathsf{info}, \{\mathsf{audit}_\ell\}_{ V_\ell\in\mathcal{V}_\mathsf{succ}} )$ then
		\item[] $\mathrm{d}_1\big(\mbf{Result}(\mathsf{info}),F(\mathsf{option}_{i_1}\ldots,\mathsf{option}_{i_n})\big) \geq d$, 
		\item[5.b.] $\bot \leftarrow \mathcal{E}(\mathsf{info}, \{\mathsf{audit}_\ell\}_{ V_\ell\in\mathcal{V}_\mathsf{succ}} )$.  
	\end{itemize}
\end{boxfig}
\begin{definition}[E2E Verifiability]\label{def:int.dfnt}
	Let $0<\epsilon<1$ and $n,m,N_v,N_b,N_t\in\mathbb{N}$ polynomial in $\lambda$ with $\theta\leq n$. 
	Let $\Pi$ be an e-voting system with $n$ voters, $N_v$ VC nodes, $N_b$ BB nodes and $N_t$ trustees.
	We say that $\Pi$  achieves \emph{end-to-end verifiability} 
		with error $\epsilon$,
	w.r.t. the election  function $F$,
	a number of $\theta$ honest successfull voters and tally deviation $d$
	if there exists a 
	(not necessarily polynomial-time) vote extractor $\mathcal{E}$  such that 
	for any PPT adversary $\A$ it holds that
	\[\Pr[G_{\mathsf{e2e\textrm{-}ver}}^{\mathcal{A},\mathcal{E},d,\theta}(1^\la,m,n,N_v,N_b,N_t)=1] \leq \epsilon.\]
\end{definition}
\vspace{5pt}
\begin{theorem} Let $n,m,N_v,N_b,N_t,\theta, d\in\mathbb{N}$ where $1\leq\theta\leq n$.
Then, D-DEMOS run with $n$ voters, $m$ options, $N_v$ VC nodes, $N_b$ BB nodes and $N_t$ trustees achieves end-to-end  
		with error $2^{-\theta}+2^{-d}$,
	w.r.t. the election  function $F$,
	a number of $\theta$ honest successfull voters and tally deviation $d$.
\end{theorem}

\begin{proof}
  Without loss of generality we can assume that every party can read consistently the data published in the majority of the BB nodes, as otherwise the adversary fails to satisfy either conditions 1 or 2 of the E2E verifiability game.
 \par We first construct a vote extractor $\mc{E}$ for D-DEMOS as follows:
 $\mc{E}$ takes input as the election transcript, $\mathsf{info}$ and a set of audit information $\set{\mathsf{audit}_\ell}_{V_\ell\in\mathcal{V}_\mathsf{succ}}$. If $\mathsf{info}$ is not meaningful, then $\mc{E}$ outputs $\bot$. 
 Let $B\leq|\tc{V}|$ be the number of different serial numbers that appear in $\set{\mathsf{audit}_\ell}_{V_\ell\in\tc{V}}$.
 Otherwise, $\mc{E}$ (arbitrarily) arranges the voters in $V_\ell\in\mathcal{V}_\mathsf{succ}$ and the serial numbers not included in
  $\set{\mathsf{audit}_\ell}_{V_\ell\in\mathcal{V}_\mathsf{succ}}
 $ as $\langle V^{\mc{E}}_{\ell}\rangle_{\ell\in[n-|\mathcal{V}_\mathsf{succ}|]}$ and $\langle\mr{tag}^{\mc{E}}_{\ell}\rangle_{\ell\in[n-B]}$ respectively. Next, for every  $\ell\in[n-|\mathcal{V}_\mathsf{succ}|]$, $\mc{E}$ extracts $\mathsf{option}_{i_\ell}$ by brute force opening and decrypting (in superpolynomial time) all the committed and encrypted BB data, or sets $\mathsf{option}_{i_\ell}$ as the zero vector, in case $V_\ell$'s vote is not published in the BB. Finally, $\mc{E}$ outputs $\langle\mathsf{option}_{i_\ell}\rangle_{ V_\ell\notin\mathcal{V}_\mathsf{succ}}$.\vspace{2pt}
%
\par We will prove the E2E verifiability of D-DEMOS based on $\mc{E}$. Assume an adversary $\A$ that wins the game $G_{\mathsf{e2e\textrm{-}ver}}^{\mathcal{A},\mathcal{E},d,\theta}(1^\la,m,n,N_v,N_b,N_t)$. Namely, $\A$ breaks E2E verifiability by allowing at least $\theta$ honest successful voters and achieving tally deviation $d$. \vspace{2pt}
\par Let $Z$ be the event that $\A$ attacks by making at least one of the option-encoding commitments associated with some cast vote-code invalid (i.e., it is in tally set $\mathbf{E}_{\mathsf{tally}}$ but it is not a commitment to some candidate encoding). Since there are at least 
$\theta$ honest and succesful voters, the min-entropy of the collected voters' coins is at least $\theta$. By min-entropy Schwartz-Zippel Lemma~\cite[Lemma 1]{DEMOS} and following the lines of~\cite[Theorem 2]{DEMOS}, we have that when the challenge is extracted from voters' coins of min-entropy $\theta$, the verification of the Chaum-Pedersen zero-knowledge proofs used in D-DEMOS for committed ballot correctness in the BB is sound except for some probability error $2^{-\theta}$.
Therefore, since there is at least one honest voter that verifies, we have that
\begin{equation}\label{eq:E2E1}
\Pr[G_{\mathsf{e2e\textrm{-}ver}}^{\mathcal{A},\mathcal{E},d,\theta}(1^\la,m,n,N_v,N_b,N_t)=1\wedge Z]\leq 2^{-\theta}\;.
\end{equation}
Now assume that $Z$ does not occur. In this case, the vote extractor $\mc{E}$ will output the indended adversarial votes up to permutation. Thus, the deviation from the intended result that $\A$ achieves, derives only by miscounting the honest votes. This may be achieved by $\A$ in two different possible ways:
\begin{itemize}
 \item \textbf{Modification attacks. } When committing to the information of some honest voter's ballot part $\A$ changes the vote-code and option
correspondence that is printed in the ballot. This attack will be detected if the voter does chooses to audit with the modified ballot part (it uses the other part to vote). The maximum deviation achieved by this attack is $1$ (the vote will count for another candidate).
\item \textbf{Clash attacks. } $\A$ provides $y$ honest voters with ballots that have the same serial number, so that the adversary can inject $y-1$ votes of his preference in the $y-1$ ``empty'' audit locations in the BB. This attack is successful only if all the $y$ voters verify the same ballot on the BB and hence miss the injected votes that produce the tally deviation. The maximum deviation achieved by this attack is $y-1$.
\end{itemize}
We stress that if $Z$ does not occur, then the above two attacks are the only meaningful\footnote{By meaningful we mean that the attack is not trivially detected. For example, the adversary may post malformed information in the BB nodes but if so, it will certainly fail at verification.} for $\A$ to follow. Indeed, if (i) all zero knowledge proofs are valid, (ii) all the honest voters are pointed to a unique audit BB location indexed by the serial number on their ballots, and (iii) the the committed in this BB location  information matches the vote-code and option association in the voters' unused ballot parts, then by the binding property of the commitments, all the tally computed by the commitments included in $\mathbf{E}_{\mathsf{tally}}$ will decrypt to the actual intended result.
\par Since the honest voters choose the used ballot parts at random, the success probability of $x$ deviation via the modification attack is $(1/2)^x$. In addition,  the success probability to clash $y$ honest voters is $(1/2)^{y-1}$ (all $y$ honest voters choose the same version to vote). As a result, by combinations of modification and clash attacks, $\A$'s success probability reduces by a factor $1/2$ for every unit increase of tally deviation. Therefore, the upper bound of the success probability of $\A$ when $Z$ does not occur is
\begin{equation}\label{eq:E2E2}
 \Pr[G_{\mathsf{e2e\textrm{-}ver}}^{\mathcal{A},\mathcal{E},d,\theta}(1^\la,m,n,N_v,N_b,N_t)=1\mid\neg Z]\leq 2^{-d}\;.
\end{equation}
By Eq.~\eqref{eq:E2E1},~\eqref{eq:E2E2}, we have that
\begin{equation*}
 \Pr[G_{\mathsf{e2e\textrm{-}ver}}^{\mathcal{A},\mathcal{E},d,\theta}(1^\la,m,n,N_v,N_b,N_t)=1]\leq 2^{-\theta}+2^{-d}\;.
\end{equation*}
\end{proof}	

\subsection{Voter Privacy}\label{subsec:sec_priv}
%
\begin{boxfig}{\label{fig:priv.game} The Voter privacy Game between the adversary $\mathcal{A}$ and the challenger $\mathsf{Ch}$ using the simuator $\mathcal{S}$.}{}
  {\it \underline{Voter privacy Game $G_{\mathsf{priv}}^{\mathcal{A},\mathcal{S},\phi}(1^{\la},n,m,N_v,N_b,N_t)$}}
  \begin{enumerate}
  \item $\mathcal{A}$ on input $1^\lambda, n, m,k$, chooses a list of options $\mathcal{P}=\{P_1,\ldots,P_m\}$, a set of voters $\mathcal{V}=\{V_1,\ldots,V_n\}$, a set of trustees $\mathcal{T}=\{T_1,\ldots,V_k\}$, a set of VC nodes $\{\mathsf{VC}_1,\ldots,\mathsf{VC}_{N_v}\}$ a set of 
  BB nodes $\{\mathsf{BB}_1,\ldots,\mathsf{BB}_{N_t}\}$.
 It provides $\mathsf{Ch}$ with all the above sets. 
 \par Throughout the game, $\mathcal{A}$ corrupts all the VC nodes a fixed set of $f_b<N_b/3$ BB nodes and a fixed set of $f_t<N_t/3$ trustees. On the other hand, $\mathsf{Ch}$ plays the role of the EA and all the non-corrupted nodes.
  \item $\mathsf{Ch}$ engages with $\mathcal{A}$ to prepare the election following the \emph{Election Authority} protocol. 
  \item After that, $\mathsf{Ch}$ chooses a bit value $b\in\{0,1\}$.
\item  The adversary $\mathcal{A}$ and the 
challenger $\mathsf{Ch}$ engage in an interaction
where $\mathcal{A}$ schedules the voters
which may run concurrently. For each voter $V_\ell\in \mathcal{V}$, the adversary chooses whether $V_\ell$ is corrupted:
\begin{itemize}
\item 
If $V_\ell$ is corrupted, then $\mathsf{Ch}$ provides the credential $s_\ell$ to  $\mathcal{A}$, who will play the
role of $V_\ell$ to cast the ballot. 
\item If $V_\ell$ is not corrupted, then $\mathcal{A}$ provides two
option selections $\langle \mathsf{option}_{\ell}^0, \mathsf{option}_{\ell}^1 \rangle$ to 
the challenger $\mathsf{Ch}$ which
operates on $V_\ell$'s behalf, voting for option $\mathsf{option}_{\ell}^b$.
 The adversary
$\mathcal{A}$ is allowed to observe the network trace.
After cast a ballot,
the challenger $\mathsf{Ch}$ provides to  $\mathcal{A}$: (i) the audit information $\alpha_\ell$ that $V_\ell$ obtains from the protocol, and 
(ii) \underline{if $b=0$}, the current view of the internal state of the voter $V_\ell$, $view_{\ell}$, that the challenger obtains during voting, or  \underline{if $b=1$}, a simulated view of the internal state of $V_\ell$ produced by $\mathcal{S}(view_{\ell})$.
\end{itemize}

\item   The adversary $\mathcal{A}$ and the 
challenger $\mathsf{Ch}$ produce the election tally, running the \emph{Trustee} protocol. $\mathcal{A}$ is allowed to observe the network trace of that protocol. 

\item  Finally, $\mathcal{A}$ using all information collected above (including the contents of the BB) outputs a bit $b^*$.
\end{enumerate}   

Denote the set of corrupted voters as $\mathcal{V}_\mathsf{corr}$ and the set of honest
voters as $\tilde{\mathcal{V}} = \mathcal{V} \setminus \mathcal{V}_\mathsf{corr}$.
The game returns a bit which is $1$ if and
    only if the following hold true:
    \begin{enumerate}[(i).]
     \item $b=b^*$ (i.e., the adversary guesses $b$ correctly).
     \item $|\mathcal{V}_\mathsf{corr}| \leq \phi$ (i.e., the number of corrupted voters is bounded by $\phi$).
     \item $f( \langle \mathsf{option}_{\ell}^0 \rangle_{V_\ell \in \tilde{\mathcal{V}}} ) = f( \langle \mathsf{option}_{\ell}^1 \rangle_{V_\ell \in \tilde{\mathcal{V}}} )$ (i.e., the election result w.r.t. the set of voters in $\tilde{\mathcal{V}}$ does not leak $b$). 
    \end{enumerate}
\end{boxfig}
Our privacy definition is extends the one used in~\cite{DEMOS} to the distributed setting of D-DEMOS. Similarly, voter privacy is defined via a \emph{Voter Privacy} game as depicted in Figure~\ref{fig:priv.game}. The game, denoted as $G_{\mathsf{priv}}^{\mathcal{A},\mathcal{S},\phi}(1^{\la},n,m,N_v,N_b,N_t)$, is played between an adversary and a challenger, and we say the adversary wins the game if and only if it returns $1$. During the game, the adversary $\mathcal{A}$ first chooses a list of options, a set of voters, a set of VC nodes, a set of BB nodes, and a set of trustees, of size $m,n,N_v,N_b,N_t$ respectively. We require that the EA is destroyed after setup, whereas the adversary may control the entire VC subsystem, up to $f_b<N_b/2$ BB nodes and up to $f_t<N_t/3$ trustees. The adversary may also corrupt up to $\phi$ voters. The adversary then instructs the honest voters to vote according to either one of two alternative ways under the restriction that election tally is the same for both ways. The system achieves voter privacy if the adversary cannot distinguish which alternative was followed by the honest voters.

\begin{definition}[Voter Privacy]\label{def:priv}
Let $0<\epsilon<1$ and $n,m,N_v,N_b,N_t\in \mathbb{N}$. Let $\Pi$ be an e-voting system with $n$ voters, $m$ options awith $n$ voters, $N_v$ VC nodes, $N_b$ BB nodes and $N_t$ trustees w.r.t. the election
 function $f$. We say that $\Pi$ achieves \emph{voter privacy} with error $\epsilon$
 for at most $\phi$ corrupted voters, if there is a PPT voter simulator $\mathcal{S}$ such that
 for any PPT adversary $\mathcal{A}$: 
 \[\big|\Pr[G_{\mathsf{priv}}^{\mathcal{A},\mathcal{S},\phi}(1^{\la},n,m,N_v,N_b,N_t)=1] - 1/2\big| = \negl(\la).\]
\end{definition}

\begin{theorem}\label{thm:sec_safety}
Assume there exists a constant $c$, $0<c<1$ such that for any $2^{\lambda^c}$-time adversary $\mathcal{A}$, the advantage of breaking the hiding property of the underlying commitment scheme is $\mathsf{Adv}_{\mathsf{hide}}(\mathcal{A}) = \mathsf{negl}(\lambda)$. Let $\phi=\lambda^{c'}$ for any constant $c'<c$. 
Then, D-DEMOS run with $n$ voters, $m$ options, $N_v$ VC nodes, $N_b$ BB nodes and $N_t$ trustees achieves voter privacy for at most $\phi$ corrupted voters.
\end{theorem}

\begin{proof}
To prove voter privacy, we explicitly construct a simulator $\mathcal{S}$ such that we can convert any adversary $\mathcal{A}$ who can win the privacy game $G_{\mathsf{priv}}^{\mathcal{A},\mathcal{S},\phi}(1^{\la},n,m,k)$ with a non-negligible probability to an adversary $\mathcal{B}$ who can break the hiding assumption of the underlying commitment scheme within $poly(\lambda)\cdot 2^{\lambda^{c'}} << 2^{\lambda^{c}}$ time.

Note that the challenger $\mathsf{Ch}$ is maintaining a coin $b\in\{0,1\}$ and always uses the option $\mathsf{option}_{\ell}^b$ to cast the honest voters' ballots. When $n-\phi <2$, the simulator $\mathcal{S}$ simply outputs the real voters' views. When $n-\phi \geq 2$, consider the following simulator $\mathcal{S}$. At the beginning of the experiment, $\mathcal{S}$ flips a coin $b'\leftarrow\{0,1\}$. For each honest voter $V_\ell$, $\mathcal{S}$ switches the vote-codes for option $\mathsf{option}_{\ell}^b$ and $\mathsf{option}_{\ell}^{b'}$.
Due to full VC corruption, $\mathcal{A}$ learns all the vote-codes. However, it does not help the adversary to disdinguish the simulated view from real view as the simulator only permutes vote-codes. Moreover, we can show that if $\mathcal{A}$ distinguishes the alternative followed by honest votes, then we can construct an algorithm $\mathcal{B}$ that invokes $\mathcal{A}$ and simulates an election execution where it guesses (i) the corrupted voters' coins (in $2^\phi$ expected attempts) and (ii) the election tally (in ${(n+1)^{m}}$ expected attempts). Thus, $\mathcal{B}$ finishes a compete simulation with high probability running in ${n^2(n+1)^{m}}\cdot2^\phi=O(2^{\lambda^{c'}})$ steps. Namely, $\mathcal{B}$ can replace all the commitments on the BB to commitments of $0$, except for the commitments in one honest voter's ballot, which commits to the guessed tally result.  By exploiting the distinguishing advantage of $\mathcal{A}$, $\mathcal{B}$ can break the hiding property of the option-encoding commitment scheme in $O(2^{\lambda^{c'}})=o(2^{\lambda^c})$ steps, thus leading to contradiction.
\end{proof}


\section{Implementation and evaluation}
\par\noindent\textbf{\emph{Implementation.}}
We implement the Election Authority component of our system as a standalone C++ application, and all other components in Java.
Whenever we store data structures on disk, or transmit them over the wire, we use Google Protocol Buffers~\cite{protobuf} to encode and decode them efficiently. We use the MIRACL library~\cite{MIRACL} for elliptic-curve cryptographic operations.
In all applications requiring a database, we use the PostgreSQL relational database system~\cite{PostgreSQL}.
\\\indent 
We build an \emph{asynchronous communications stack} (ACS) on top of Java, using Netty~\cite{Netty} and the asynchronous PostgreSQL driver from~\cite{PGASYNC}, using TLS based authenticated channels for inter-node communication, and a public HTTP channel for public access.
This infrastructure uses connection-oriented sockets, but allows the applications running on the upper layers to operate in a message-oriented fashion.
We use this infrastructure to implement VC and BB nodes. 
We implement Bracha's Binary Consensus directly on top of the ACS, and we use that to implement our Vote Set Consensus algorithm. 
We introduce a version of Binary Consensus that operates in batches of arbitrary size; this way, we achieve greater network efficiency.
We implement ``verifiable secret sharing with honest dealer'', by utilizing Shamir's Secret Share library implementation~\cite{JavaShamir}, and having the EA sign each share.
\ifextended
\par\noindent
\textbf{Web browser replicated service reader.} 
Our choice to model the Bulletin Board as a replicated service of non-cooperating nodes puts the burden of response verification on the reader of the service; a human reader is expected to manually issue a read request to all nodes, then compare the responses and pick the one posted by the majority of nodes. 
To alleviate this burden, we implement a web browser extension which automates this task, as an extension for Mozilla Firefox. The user sets up the list of URLs for the replicated service.  The add-on 1) intercepts any HTTP request towards any of these URLs, 2) issues the same request to the rest of the nodes, and 3) captures the responses, compares them in binary form, and routes the response coming from the majority, as a response to the original request posted by the user. Majority is defined by the number of defined URL prefixes; for 3 such URLs, the first 2 equal replies suffice. 
\\\indent 
With the above approach, the user never sees a wrong reply, as it is filtered out by the extension. 
Also note this process will be repeated for all dependencies of the initial web page (images, scripts, CSS), as long at they come from the same source (with the same URL prefix), verifying the complete user visual experience in the browser.  
\fi

\par\noindent\textbf{\emph{Evaluation.}}
\label{sect:evaluation}
\begin{figure*}[!ht]
\centering
{
  \subfloat[]
  {
    \includegraphics[width=0.33\textwidth]{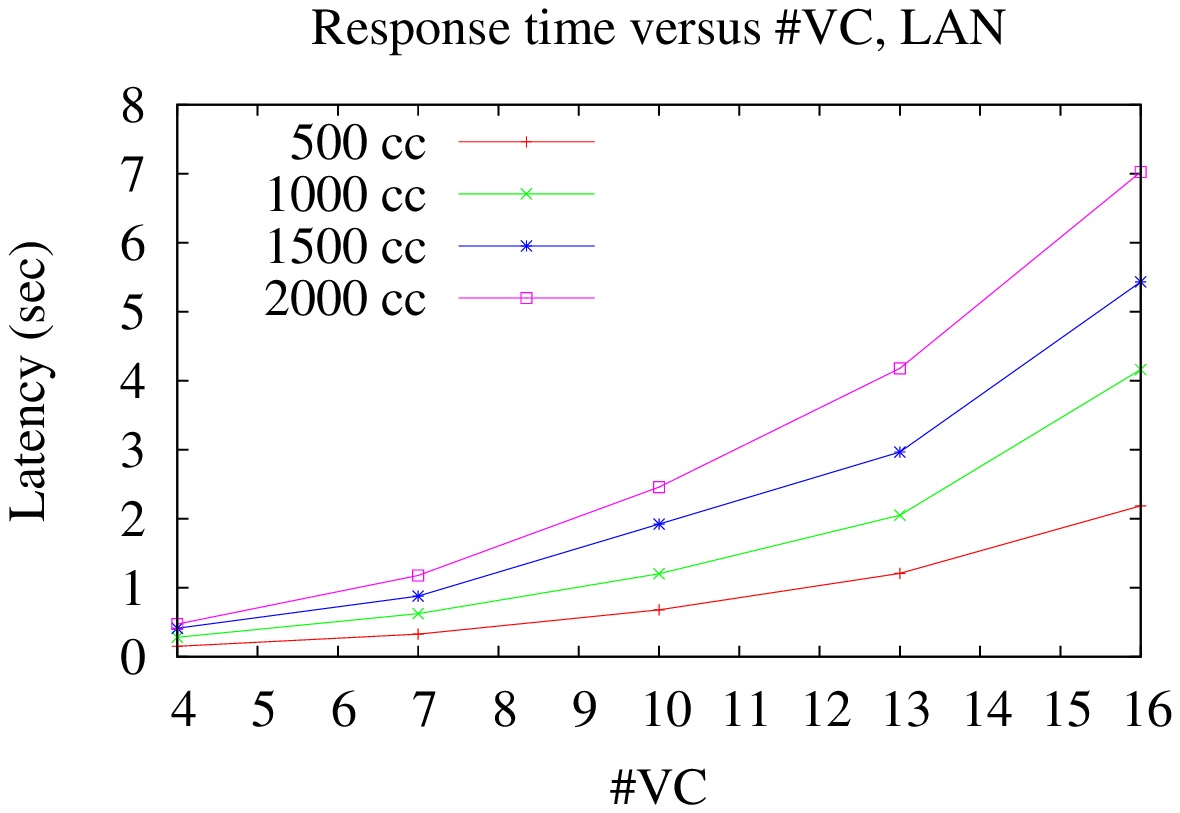}
    \label{fig:LANResponseTimeVC}
  }
  \subfloat[]
  {
    \includegraphics[width=0.33\textwidth]{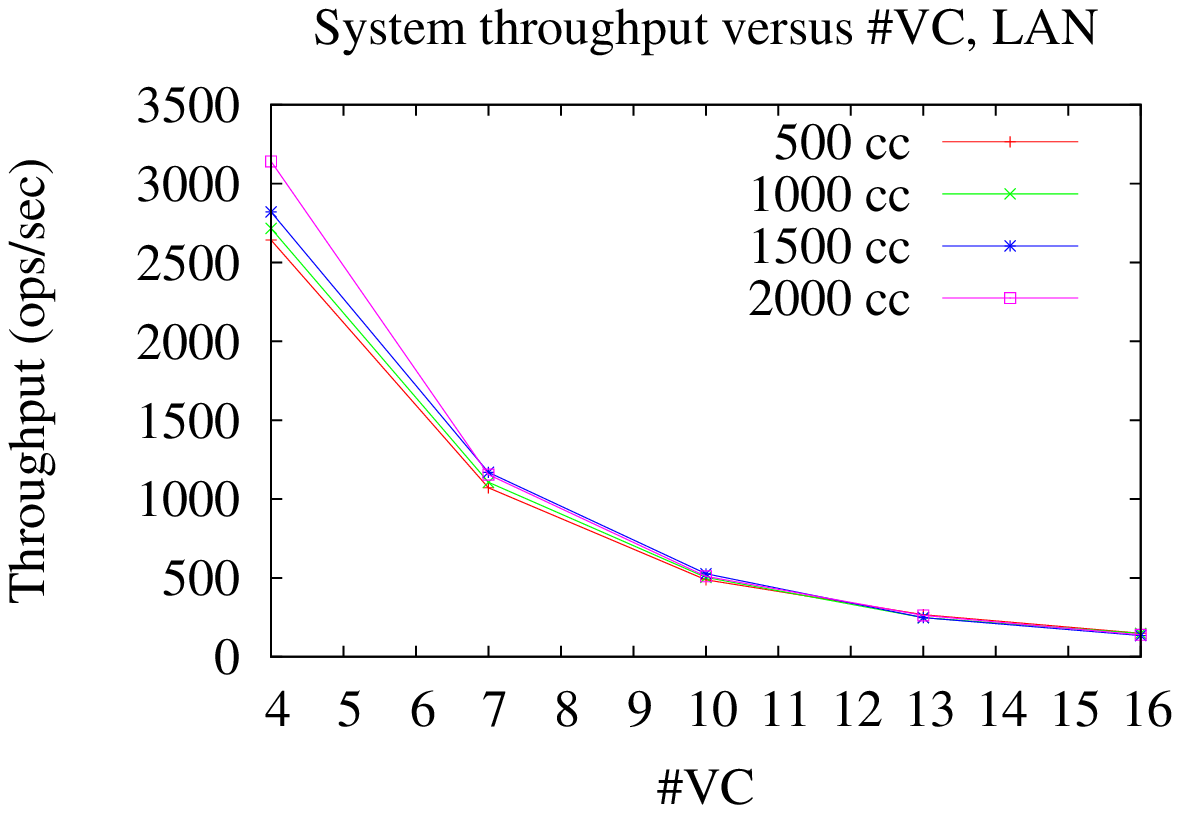}
    \label{fig:LANThroughputVC}
  }
  \subfloat[]
  {
    \includegraphics[width=0.33\textwidth]{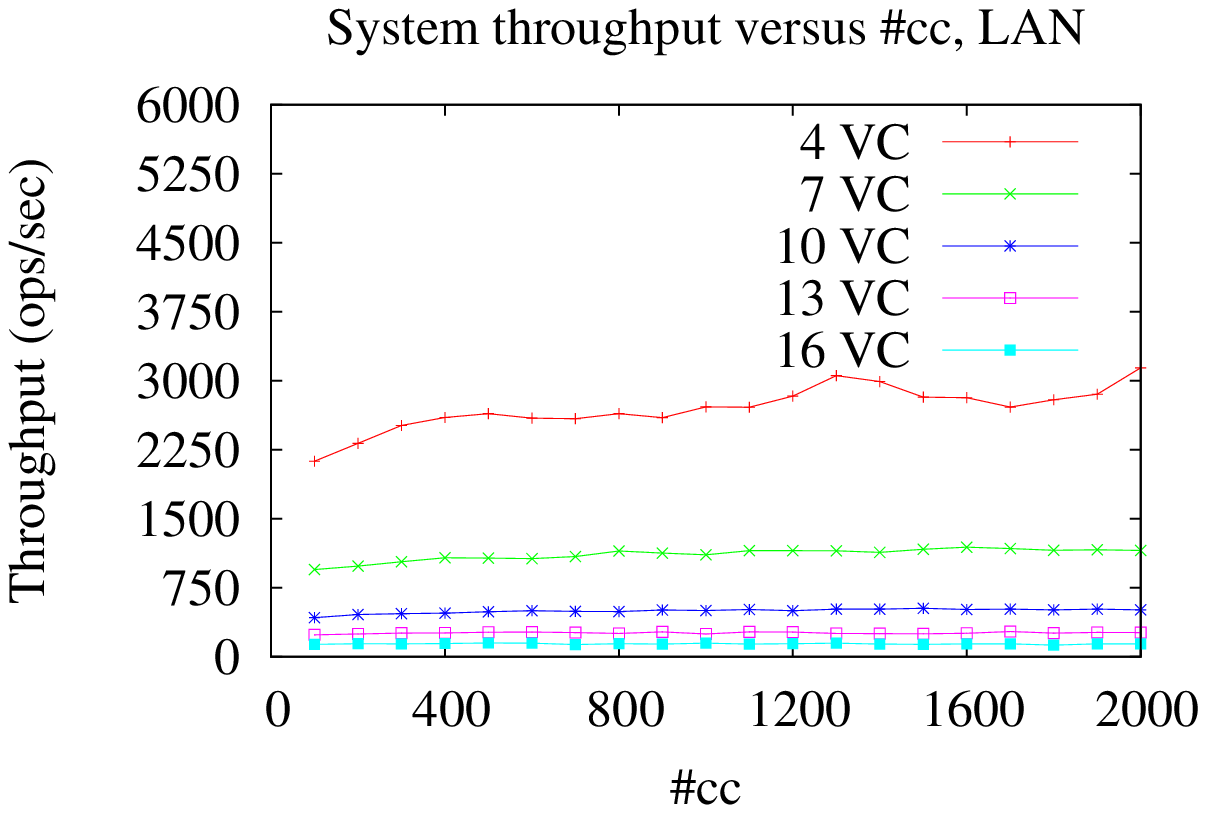}
    \label{fig:LANThroughputCC}
  }
  \hfil
  \subfloat[]
  {
    \includegraphics[width=0.33\textwidth]{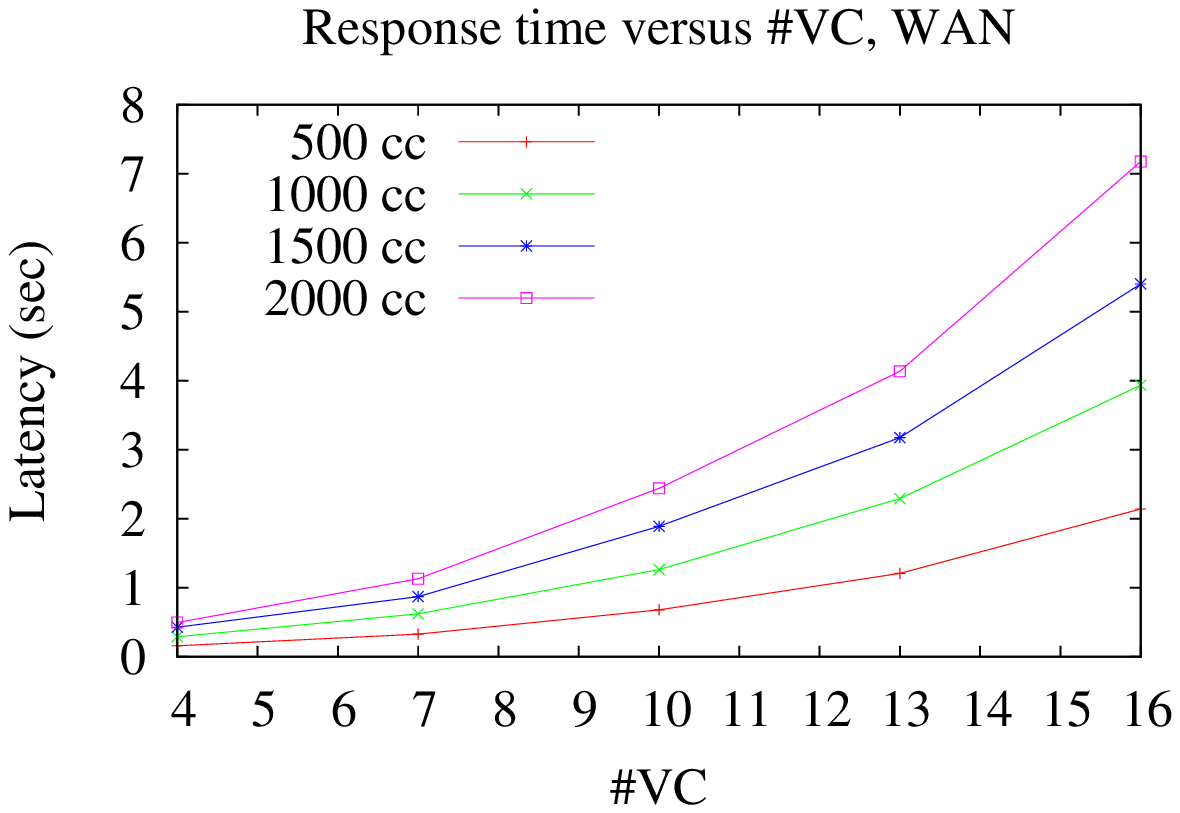}
    \label{fig:WANResponseTimeVC}
  }
  \subfloat[]
  {
    \includegraphics[width=0.33\textwidth]{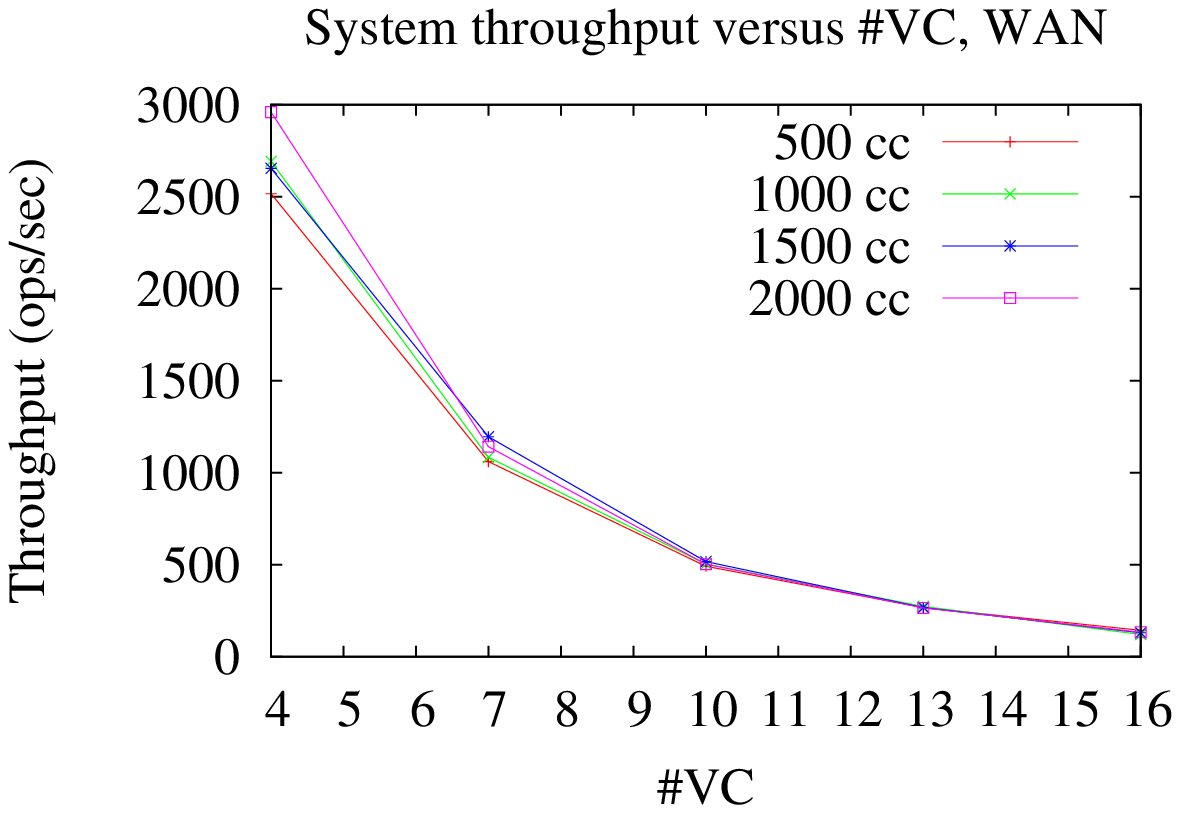}
    \label{fig:WANThroughputVC}
  }
  \subfloat[]
  {
    \includegraphics[width=0.33\textwidth]{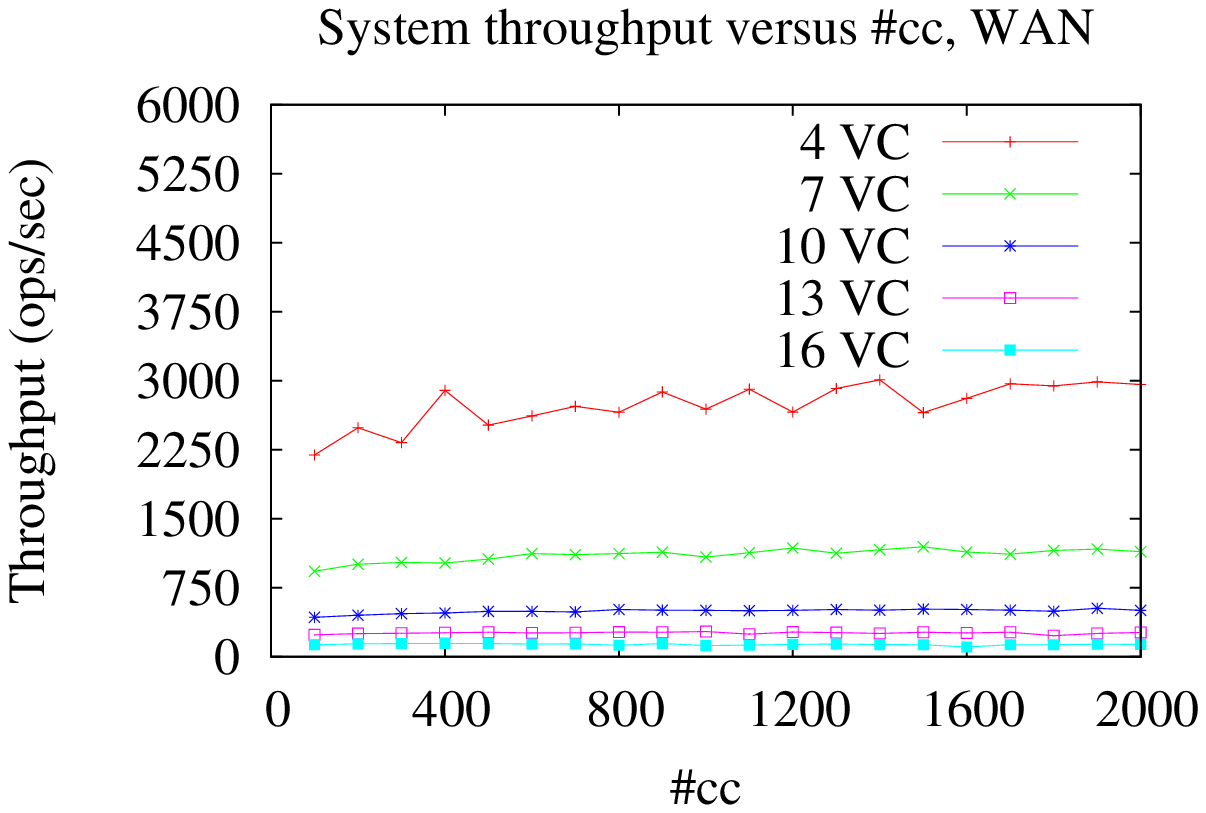}
    \label{fig:WANThroughputCC}
  } 
}
\caption{Latency (\ref{fig:LANResponseTimeVC}, \ref{fig:WANResponseTimeVC}) and throughput graphs
(\ref{fig:LANThroughputVC}, \ref{fig:WANThroughputVC}) of the vote collection algorithm vs. the number of VC nodes. Figures
(\ref{fig:LANThroughputCC} and \ref{fig:WANThroughputCC}) illustrate throughput versus the number of concurrent clients. First
row illustrates LAN setting plots. Second row illustrates WAN setting plots. Election parameters are $n$ = 200,000 and $m$ = 4.}
\label{figure:comparisonfigure}
\end{figure*}
\begin{figure*}[!ht]
\centering
{
  \subfloat[]
  {
    \includegraphics[width=0.33\textwidth]{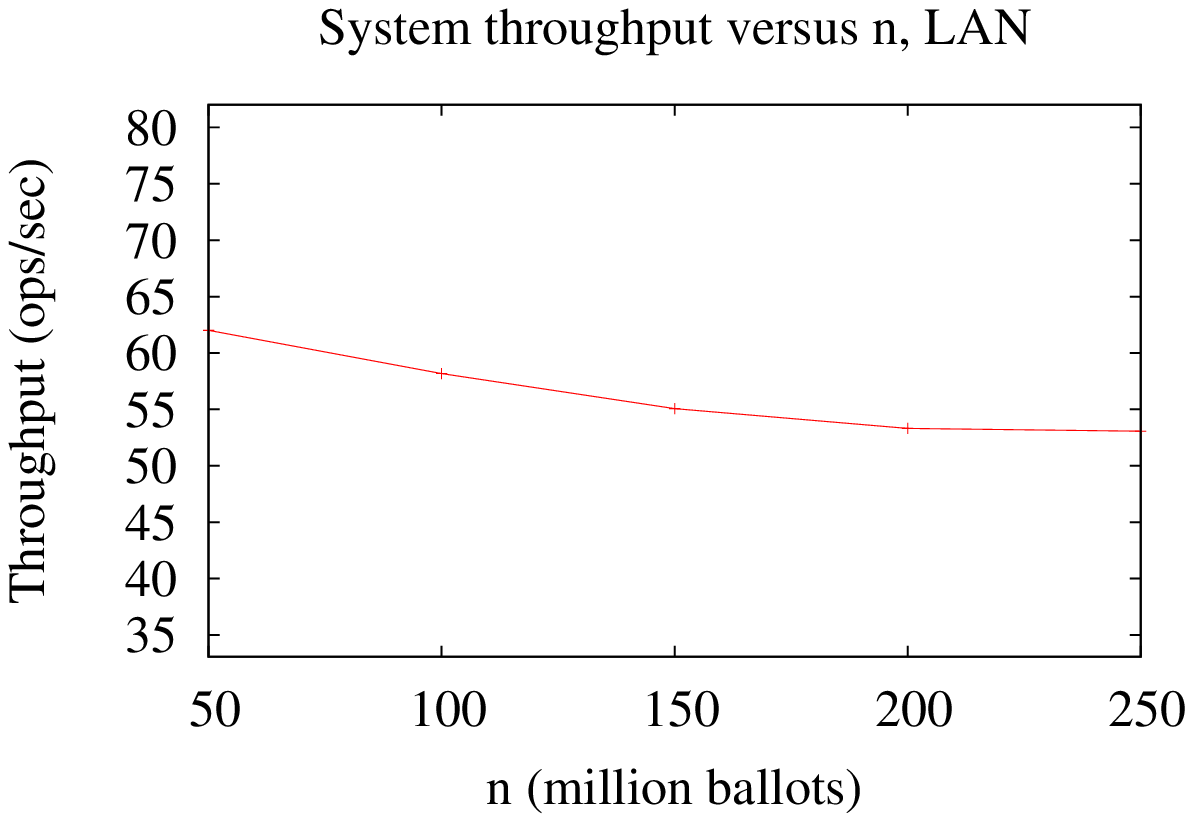}
    \label{fig:LANThroughputn}
  }
  \subfloat[]
  {
    \includegraphics[width=0.33\textwidth]{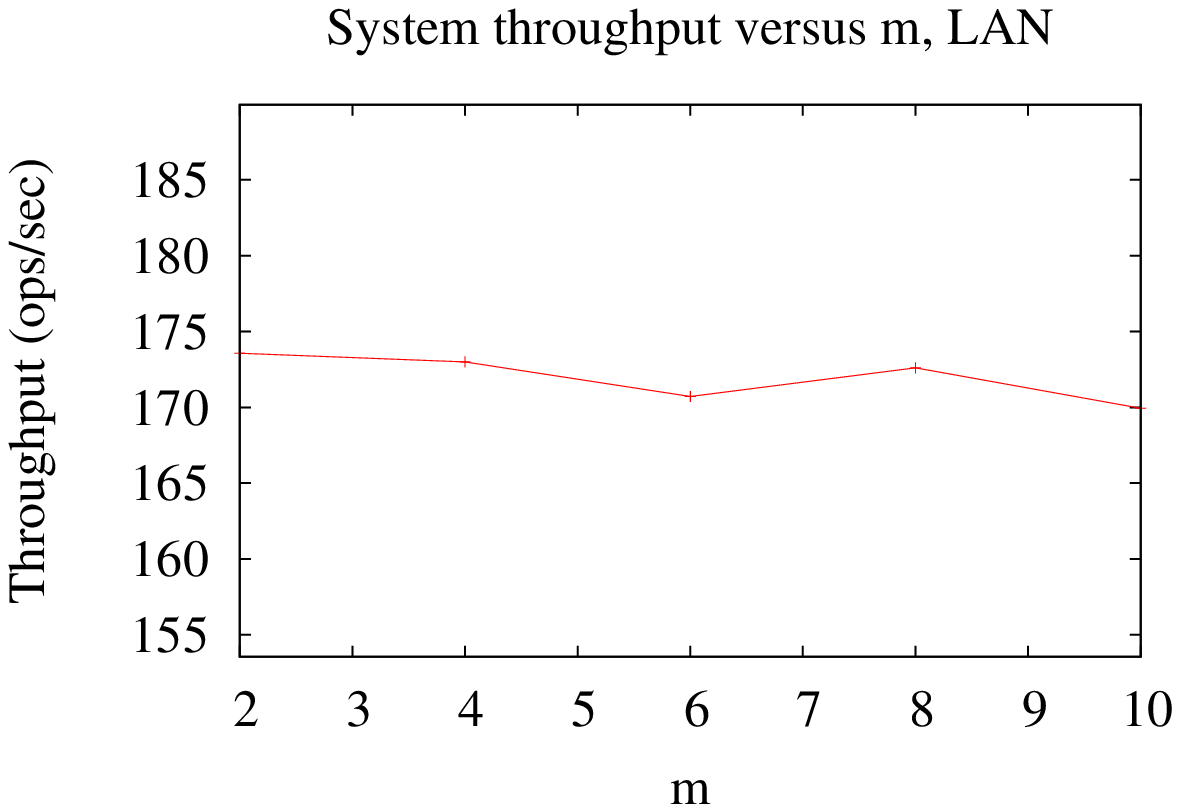}
    \label{fig:LANThroughputm}
  }
  \subfloat[]
  {
    \includegraphics[width=0.33\textwidth]{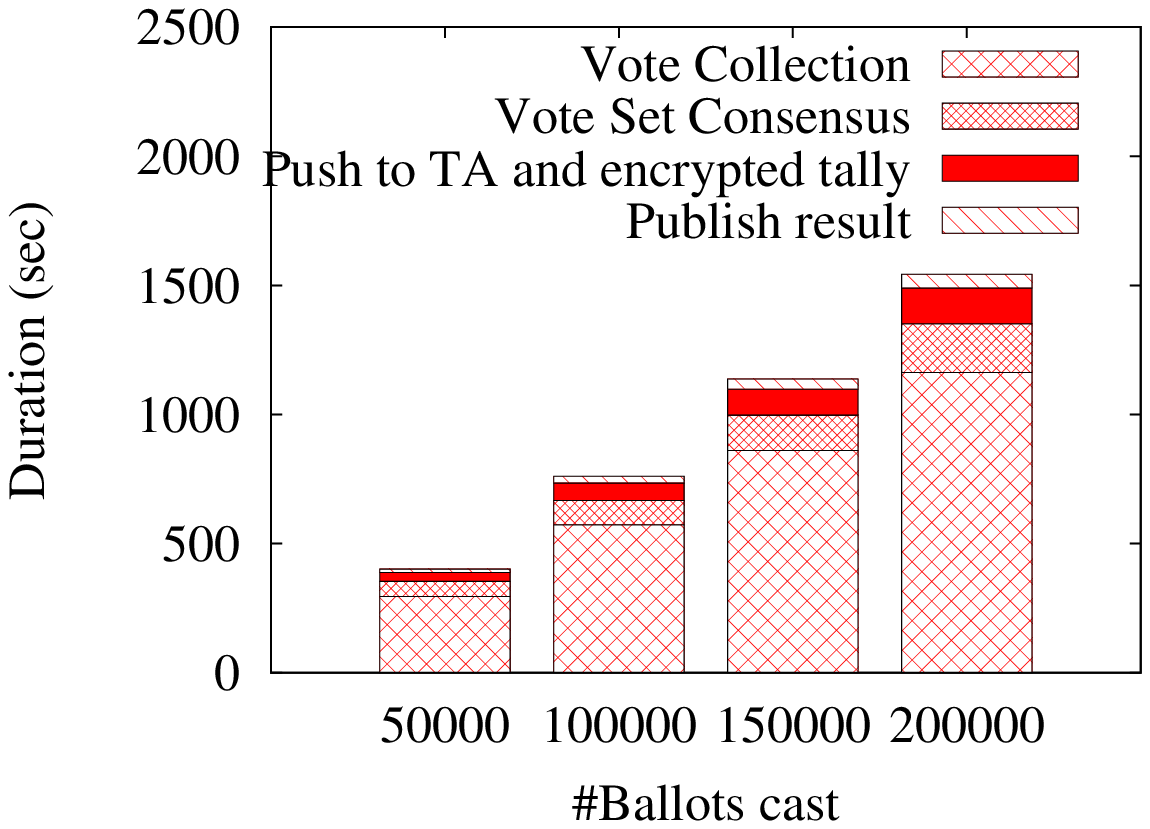}
    \label{fig:OverallLatencyBar}
  }
  \caption{Throughput graphs of the vote collection phase versus the number of total election ballots $n$ (\ref{fig:LANThroughputn}) and the number of total election options $m$ (\ref{fig:LANThroughputm}). A total of 200,000 ballots were cast by 400 concurrent clients on 4 VC nodes. Figure \ref{fig:OverallLatencyBar} illustrates the duration of all system phases. Results depicted are for 4 VCs, $n$ = 200,000 and $m$ = 4. All these plots are for disk based experiments.}
  \label{figure:placeholder}
}
\end{figure*}
We experimentally evaluate the performance of our voting system, focusing mostly on our vote collection algorithm, which is the most performance critical part. 
We conduct our experiments using a cluster of 12 machines, connected over a Gigabit Ethernet switch. 
The first 4 are equipped with Hexa-core Intel Xeon E5-2420 @ 1.90GHz, 16GB RAM, and one 1TB SATA disk, running CentOS 7 Linux, and we use them to run our VC nodes. 
The remaining 8 comprise dual Intel(R) Xeon(TM) CPUs @ 2.80GHz, with 4GB of main memory, and two 50GB disks, running CentOS 6 Linux, and we use them as clients. 
\\\indent 
We implement a multi-threaded voting client to simulate concurrency. 
It starts the requested number of threads, each of which loads its corresponding ballots from disk and waits for a signal to start; from then on, the thread enters a loop where it picks one VC node and vote code at random, requests the voting page from the selected VC (HTTP GET), submits its vote (HTTP POST), and waits for the reply (receipt). 
This simulates multiple concurrent voters casting their votes in parallel, and gives an understanding of the behavior of the system under the corresponding load. 
\\\indent
We employ the PostgreSQL RDBMS~\cite{PostgreSQL} to store all VC initialization data from the EA. 
We start off by demonstrating our system's capability of handling large-scale elections. 
To this end, we generate election data for referendums, i.e., $m=2$, and vary the total number of ballots $n$ from 50 million to 250 million (note the 2012 US voting population size was 235 million). 
We fix the number of concurrent clients to 400 and cast a total of 200,000 ballots, which are enough for our system to reach its steady-state operation. 
Figure~\ref{fig:LANThroughputn} shows the throughput of the system declines slowly, even with a five-fold increase in the number of eligible voters.
\\\indent 
In our second experiment, we explore the effect of $m$, i.e., the number of election options, on system performance.
We vary the number of options from $m=2$ to $m=10$. 
Each election has a total of $n=200,000$ ballots which we spread evenly across 400 concurrent clients. 
As illustrated in Figure \ref{fig:LANThroughputm}, our vote collection protocol
manages to deliver approximately the same throughput regardless of the value of $m$. 
Notice that the only extra overhead $m$ induces during vote collection, is the increase in the number of hash verifications during vote code validation, as there are more vote codes per ballot.
\\\indent 
Next, we evaluate the scalability of our vote collection protocol by varying the number of vote collectors and concurrent clients.
We eliminate the database, by caching the election data in memory and servicing voters from the cache, to measure the net communication and processing costs of our voting protocol.
We vary the number of VC nodes from 4 to 16, and distribute them across the 4 physical machines. 
Note that, co-located nodes are unable to produce vote receipts via local messages only, since the $N_{v}-f_{v}$ threshold cannot be satisfied, i.e., cross-machine communication is still the dominant factor in receipt generation.
For election data, we use the dataset with  $n=200,000$ ballots and $m=4$ options. 
\\\indent 
In Figures \ref{fig:LANResponseTimeVC} and \ref{fig:LANThroughputVC}, we plot the average response time and throughput of our vote collection protocol, versus the number of vote collectors, under various concurrent client scenarios. 
Results illustrate an almost linear increase in the client-perceived latency, for all concurrency scenarios, up to 13 VC nodes. 
From this point on, when four logical VC nodes are placed on a single physical machine, we notice a non-linear increase in latency.
We attribute this to the overloading of the memory bus, a resource shared among all processors of the system, which services all (in-memory) database operations.
\\\indent In terms of overall system throughput, however, the penalty of tolerating extra failures, i.e., increasing the number of vote collectors, manifests early on. 
We notice an almost 50\% decline in system throughput from 4 to 7 VC nodes. 
However, further increases in the number of vote collectors lead to a much smoother, linear decrease. 
We repeat the same experiment by emulating a WAN environment using \emph{netem}~\cite{Netem}, a network emulator for Linux.
We inject a uniform latency of 25ms (typical for US coast-to-coast communication~\cite{coast2coast}) for each network packet exchanged between vote collector nodes, and present our results in Figures \ref{fig:WANResponseTimeVC} and \ref{fig:WANThroughputVC}.
A simple comparison between LAN and WAN plots illustrates our system manages to deliver the same level of throughput and average response time, regardless of the increased intra-VC communication latency.
Finally, in Figures \ref{fig:LANThroughputCC} and \ref{fig:WANThroughputCC}, we plot system throughput versus the number of concurrent clients, in LAN and WAN settings respectively. 
Results show our system has the nice property of delivering nearly constant throughput, regardless of the incoming request load, for a given number of VC nodes.
\\\indent 
Finally, in Figure \ref{fig:OverallLatencyBar}, we illustrate a breakdown of the duration of each phase of the complete voting system (D-DEMOS), versus the total number of ballots cast. 
We assume immediate phase succession, i.e., the vote collection phase ends when all votes have been cast, at which point the vote set consensus phase starts, and so on. 
The ``Push to BB and encrypted tally'' phase is the time it takes for the vote collectors to push the final vote code set to the BB nodes, including all actions necessary by the BB to calculate and publish the encrypted result.  
The ``Publish result'' phase is the time it takes for Trustees to calculate and push their share of the opening of the  final tally to the BB, and for the BB to publish the final tally.   
Note that, in most voting procedures, the vote collection phase would in reality last several hours and even days as stipulated by national law (see Estonia voting system). 
Thus, looking only at the post-election phases of the system, we see that the time it takes to publish the final tally on the BB is quite fast. 
\\\indent 
Overall, although we introduced Byzantine Fault Tolerance across all phases of a voting system (besides setup), we demonstrate it achieves high performance, enough to run real-life elections of large electorate bodies.

\section{Conclusion and future work}
We have presented the world's first complete, state-of-the-art, end-to-end verifiable, distributed voting system with no single point of failure besides setup. 
The system allows voters to verify their vote was tallied-as-intended without the assistance of special software or trusted devices, and external auditors to verify the correctness of the election process. 
Additionally, the system allows voters to delegate auditing to a third party auditor, without sacrificing their privacy.
We provided a model and security analysis of our voting system.
Finally, we implemented a prototype of the integrated system, measured its performance and demonstrated its ability to handle large scale elections.  
\\\indent
As future work, we plan to expand our system to \emph{k-out-of-m} elections.


\Urlmuskip=0mu plus 1mu\relax
\bibliographystyle{IEEEtranS}
\bibliography{MyBibFile}

\begin{thebibliography}{10}
\providecommand{\url}[1]{#1}
\csname url@samestyle\endcsname
\providecommand{\newblock}{\relax}
\providecommand{\bibinfo}[2]{#2}
\providecommand{\BIBentrySTDinterwordspacing}{\spaceskip=0pt\relax}
\providecommand{\BIBentryALTinterwordstretchfactor}{4}
\providecommand{\BIBentryALTinterwordspacing}{\spaceskip=\fontdimen2\font plus
\BIBentryALTinterwordstretchfactor\fontdimen3\font minus
  \fontdimen4\font\relax}
\providecommand{\BIBforeignlanguage}[2]{{%
\expandafter\ifx\csname l@#1\endcsname\relax
\typeout{** WARNING: IEEEtranS.bst: No hyphenation pattern has been}%
\typeout{** loaded for the language `#1'. Using the pattern for}%
\typeout{** the default language instead.}%
\else
\language=\csname l@#1\endcsname
\fi
#2}}
\providecommand{\BIBdecl}{\relax}
\BIBdecl

\bibitem{PGASYNC}
``Asynchronous postgresql java driver.''
  \url{https://github.com/alaisi/postgres-async-driver/}.

\bibitem{protobuf}
``Google protocol buffers,'' \url{https://code.google.com/p/protobuf/}.

\bibitem{coast2coast}
``High performance browser networking: What every web developer should know
  about networking and web performance,''
  \url{http://chimera.labs.oreilly.com/books/1230000000545/ch01.html#PROPAGATION_LATENCY}.

\bibitem{MIRACL}
``Miracl multi-precision integer and rational arithmetic c/c++ library.''
  \url{http://www.certivox.com/miracl/}.

\bibitem{Netty}
``Netty, an asynchronous event-driven network application framework.''
  \url{http://netty.io/}.

\bibitem{PostgreSQL}
``Postgresql rdbms.'' \url{http://www.postgresql.org/}.

\bibitem{JavaShamir}
``Shamir's secret share in java.''
  \url{https://github.com/timtiemens/secretshare}.

\bibitem{adida-helios-2008}
B.~Adida, ``Helios: Web-based open-audit voting,'' in \emph{USENIX Security
  Symposium}, 2008.

\bibitem{aublin2013rbft}
P.-L. Aublin, S.~Ben~Mokhtar, and V.~Qu{\'e}ma, ``Rbft: Redundant byzantine
  fault tolerance,'' in \emph{IEEE ICDCS}, 2013.

\bibitem{benaloh2013starvote}
J.~Benaloh, M.~D. Byrne, B.~Eakin, P.~T. Kortum, N.~McBurnett, O.~Pereira,
  P.~B. Stark, D.~S. Wallach, G.~Fisher, J.~Montoya, M.~Parker, and M.~Winn,
  ``{STAR}-vote: {A} secure, transparent, auditable, and reliable voting
  system,'' in \emph{{EVT/WOTE} '13}, Aug. 2013.

\bibitem{castro-osdi-1999}
M.~Castro and B.~Liskov, ``Practical byzantine fault tolerance,'' in
  \emph{OSDI}, February 1999.

\bibitem{chaum2001surevote}
D.~Chaum, ``Surevote: Technical overview,'' in \emph{Proceedings of the
  Workshop on Trustworthy Elections}, ser. WOTE, Aug. 2001.

\bibitem{chaum2008scantegrity}
D.~Chaum, A.~Essex, R.~Carback, J.~Clark, S.~Popoveniuc, A.~Sherman, and
  P.~Vora, ``Scantegrity: End-to-end voter-verifiable optical-scan voting,''
  \emph{Security \& Privacy, IEEE}, vol.~6, no.~3, pp. 40--46, 2008.

\bibitem{CP}
D.~Chaum and T.~P. Pedersen, ``Wallet databases with observers,'' in
  \emph{CRYPTO '92}.\hskip 1em plus 0.5em minus 0.4em\relax Springer-Verlag,
  1993, pp. 89--105.

\bibitem{chaum-esorics-2005}
D.~Chaum, P.~Y.~A. Ryan, and S.~A. Schneider, ``A practical voter-verifiable
  election scheme,'' in \emph{{ESORICS} 2005}, Sept. 2005, pp. 118--139.

\bibitem{clarkson2008civitas}
M.~R. Clarkson, S.~Chong, and A.~C. Myers, ``Civitas: Toward a secure voting
  system,'' in \emph{{IEEE} Symposium on Security and Privacy}, 2008.

\bibitem{clement2009upright}
A.~Clement, M.~Kapritsos, S.~Lee, Y.~Wang, L.~Alvisi, M.~Dahlin, and T.~Riche,
  ``Upright cluster services,'' in \emph{Proc. of ACM SOSP}, 2009.

\bibitem{clement2009making}
A.~Clement, E.~L. Wong, L.~Alvisi, M.~Dahlin, and M.~Marchetti, ``Making
  byzantine fault tolerant systems tolerate byzantine faults.'' in \emph{NSDI},
  vol.~9, 2009, pp. 153--168.

\bibitem{cowling2006hq}
J.~Cowling, D.~Myers, B.~Liskov, R.~Rodrigues, and L.~Shrira, ``Hq replication:
  A hybrid quorum protocol for byzantine fault tolerance,'' in
  \emph{Proceedings of USENIX OSDI}, 2006.

\bibitem{CGS-eurocrypt-1997}
R.~Cramer, R.~Gennaro, and B.~Schoenmakers, ``A secure and optimally efficient
  multi-authority election scheme,'' in \emph{{EUROCRYPT}}, 1997.

\bibitem{culnane2014vVote}
\BIBentryALTinterwordspacing
C.~Culnane, P.~Y.~A. Ryan, S.~Schneider, and V.~Teague, ``vvote: a verifiable
  voting system {(DRAFT)},'' \emph{CoRR}, vol. abs/1404.6822, 2014. [Online].
  Available: \url{http://arxiv.org/abs/1404.6822}
\BIBentrySTDinterwordspacing

\bibitem{culnane2014peered}
C.~Culnane and S.~Schneider, ``A peered bulletin board for robust use in
  verifiable voting systems,'' in \emph{Computer Security Foundations Symposium
  (CSF), 2014 IEEE 27th}.\hskip 1em plus 0.5em minus 0.4em\relax IEEE, 2014,
  pp. 169--183.

\bibitem{dini2003secure}
G.~Dini, ``A secure and available electronic voting service for a large-scale
  distributed system,'' \emph{Future Generation Computer Systems}, vol.~19,
  no.~1, pp. 69--85, 2003.

\bibitem{ElGamal}
T.~El~Gamal, ``A public key cryptosystem and a signature scheme based on
  discrete logarithms,'' in \emph{Springer-Verlag CRYPTO 1984}.

\bibitem{fisher-wote-2006}
K.~Fisher, R.~Carback, and A.~Sherman, ``Punchscan: introduction and system
  definition of a high-integrity election system,'' in \emph{{WOTE}}, 2006.

\bibitem{gjosteen2013norway}
\BIBentryALTinterwordspacing
K.~Gj{\o}steen, ``The norwegian internet voting protocol,'' \emph{{IACR}
  Cryptology ePrint Archive}, vol. 2013, p. 473, 2013. [Online]. Available:
  \url{http://eprint.iacr.org/2013/473}
\BIBentrySTDinterwordspacing

\bibitem{Netem}
S.~Hemminger \emph{et~al.}, ``Network emulation with netem,'' in \emph{Linux
  conf au}.\hskip 1em plus 0.5em minus 0.4em\relax Citeseer, 2005, pp. 18--23.

\bibitem{DEMOS}
A.~Kiayias, T.~Zacharias, and B.~Zhang, ``End-to-end verifiable elections in
  the standard model,'' in \emph{{EUROCRYPT} 2015}, April 2015, pp. 468--498.

\bibitem{kotla2007zyzzyva}
R.~Kotla, L.~Alvisi, M.~Dahlin, A.~Clement, and E.~Wong, ``Zyzzyva: Speculative
  byzantine fault tolerance,'' in \emph{SOSP}, Oct 2007.

\bibitem{kutylowski2010SCV}
\BIBentryALTinterwordspacing
M.~Kutylowski and F.~Zag{\'{o}}rski, ``Scratch, click {\&} vote: {E2E} voting
  over the internet,'' in \emph{Towards Trustworthy Elections, New Directions
  in Electronic Voting}, 2010, pp. 343--356. [Online]. Available:
  \url{http://dx.doi.org/10.1007/978-3-642-12980-3_21}
\BIBentrySTDinterwordspacing

\bibitem{Lynch:1996:DA}
N.~Lynch, \emph{Distributed Algorithms}.\hskip 1em plus 0.5em minus 0.4em\relax
  Morgan Kaufmann, 1996.

\bibitem{vss}
T.~Pedersen, ``Non-interactive and information-theoretic secure verifiable
  secret sharing,'' in \emph{Advances in Cryptology — CRYPTO}, 1991.

\bibitem{zagorski2013remotegrity}
F.~Zag{\'o}rski, R.~T. Carback, D.~Chaum, J.~Clark, A.~Essex, and P.~L. Vora,
  ``Remotegrity: Design and use of an end-to-end verifiable remote voting
  system,'' in \emph{Applied Cryptography and Network Security}, 2013.

\end{thebibliography}

\end{document}